\documentclass[11pt]{article}

\usepackage{fullpage,amsfonts,amssymb,amsthm,amsmath}
\usepackage{graphicx}
\usepackage[colorlinks, bookmarksnumbered, linkcolor=blue, citecolor=blue, urlcolor=green]{hyperref}
\usepackage{geometry}
\usepackage{natbib}
\usepackage{cases}
\usepackage{subfig}
\usepackage{float}
\usepackage{array}
\usepackage{xcolor}
\usepackage{mathtools}
\usepackage{enumerate}  
\usepackage{extarrows}
\title{On Reward Sharing in Blockchain Mining Pools}

\author{Burak Can\thanks{
		Department of Data Analytics and Digitalisation, Maastricht University, the
		Netherlands, E-mail: \texttt{b.can@maastrichtuniversity.nl}. This work was partially supported by the Graduate School of Business and Economics, Maastricht University.} \and
	Jens Leth Hougaard\thanks{Economics, NYU-Shanghai, China, E-mail: \texttt{jlh21@nyu.edu};  and IFRO, University of Copenhagen, Denmark, E-mail: \texttt{jlh@ifro.ku.dk}. This work is supported by the Center for Blockchains and Electronic Markets funded by the Carlsberg Foundation under grant no. CF18-1112.
	} \and
	Mohsen Pourpouneh\thanks{
		Department of Food and Resource Economics (IFRO), University of Copenhagen,
		Denmark, E-mail: \texttt{mohsen@ifro.ku.dk.} This work is supported by the Center for Blockchains and Electronic Markets funded by the Carlsberg Foundation under grant no. CF18-1112. } 
}

\newtheorem{theorem}{Theorem}
\newtheorem{prop}{Proposition}
\newtheorem{lemma}{Lemma}
\newtheorem{definition}{Definition}
\newtheorem{corollary}{Corollary}
\newtheorem{remark}{Remark}
\newtheorem{example}{Example}

\begin{document}
	\maketitle
	\noindent
	\rule[0.5ex]{1\columnwidth}{1pt}
	\begin{abstract}
	This paper proposes a conceptual framework for the analysis of reward sharing schemes in  mining pools, such as those associated with Bitcoin. The framework is centered around the reported shares in a pool instead of agents and results in two new fairness criteria, absolute and relative redistribution. These criteria impose that the addition of a share to the pool affects all previous shares in the same way, either in absolute amount or in relative ratio.  We characterize two large classes of economically viable reward sharing schemes corresponding to each of these fairness criteria in turn.  We further show that the intersection of these classes brings about a generalization of  the well-known proportional scheme, which also leads to a new characterization of the proportional scheme as a corollary. 
	\end{abstract}
	\textit{Keywords:} Blockchain, Bitcoin, fairness, mining pools, resource allocation, mechanism design.
	
\noindent	\textit{JEL Classification: D63,  G20,   L86, D31.}

\noindent
\rule[0.5ex]{1\columnwidth}{1pt}

\newpage

\section{Introduction}
The invention of the first decentralized cryptocurrency, Bitcoin, (\cite{nakamoto2019bitcoin}) sparked a huge interest in decentralized networks with distributed trust, both academically and businesswise. The technology behind these decentralized networks is called Blockchain. Loosely speaking, a blockchain is a ledger composed of an immutable chain of transactions organized in \emph{blocks}. Each block is synchronized across the network users (nodes) through a distributed consensus protocol which ensures that all the nodes in the network agree on the latest status of the ledger\footnote{To agree on the latest status of the ledger, Blockchains utilize various different consensus protocols, see \cite{nguyen2018survey,mingxiao2017review}.}.  

In terms of financials, as of May 2021, the global crypto market cap exceeds $2.26$ trillion\footnote{https://coinmarketcap.com} and a single Bitcoin is traded around $56000$ USD while its total market cap is comparable to the GDPs of various countries\footnote{For example, Hungary, Kenya, Luxembourg and many others. See the following link for a ranking: \url{https://worldpopulationreview.com/countries/countries-by-gdp}. For other financial statistics, see \url{https://bitinfocharts.com/bitcoin/}.}. In terms of other use cases, there are an increasing number decentralized applications on various blockchains, e.g., Ethereum mainnet has about 2700 dApps deployed, providing solutions in banking, finance, law,
logistics, and other sectors. Following the global interest in Blockchains, all EU members and European Commission have joined forces to form the European Blockchain Partnership (EBP) while China has already launched trials for its national cryptocurrency, the digital Yuan.

\subsection{Consensus Protocols and Pools}
The development of the blockchain technology still in its infancy. The oldest consensus protocol that has a proven track record is the original Proof-of-Work (PoW)\footnote{To achieve consensus under PoW protocol, some nodes, called \emph{miners}, compete with one another to solve a cryptographic puzzle. Miners search for an integer (a nonce) such that when combined (hashed) together with the list of transactions (the block header), produces another integer known as the hash value. In case, this hash value is less than a predefined number (network target value) set by the network, the so called puzzle is solved.} protocol in~\cite{nakamoto2019bitcoin}, while many new blockchains utilize Proof-of-Stake (PoS) protocol for its energy efficiency and transaction throughput. The process of achieving consensus is called \emph{mining} in PoW, while this is achieved by \emph{staking} in PoS\footnote{The solution to this aforementioned puzzle in PoW is called a \emph{full solution}, for which the successful miner is given a financial reward. Similarly participants in PoS are awarded according to their stakes in the network}. In theory, the probability of finding a full solution in PoW is proportional to the computational power of the miner\footnote{For an extensive demonstration of these concepts see Anders Brownworth's blog at  \url{https://andersbrownworth.com/blockchain/hash}.}. However, the computational power of individual miners are negligible in comparison to that of the network\footnote{At the time of writing this paper, the total hashing power of the Bitcoin blockchain is $120\time 10^{18}$KH/s, while a state-of-the-art CPU has about $11$KH/s. Therefore, the probability of finding a full solution as a solo miner is roughly $\frac{1}{10^{19}}$.}. Therefore, mining alone leads to a highly unstable income and incentivizes miners to \textit{pool} their resources and split the resulting rewards. Such cooperative actions result in lower income variation for miners (\cite{romiti2019deep,rosenfeld2011analysis}). This is why, despite being a decentralized system by design, PoW has been shown to induce centralization both theoretically (\cite{leshno2019bitcoin} and \cite{chen2019axiomatic}) and in practice through the emergence and total dominance of centralized mining-pools. In fact, as of 2021, almost  $90\%$ of the total computational power in Bitcoin blockchain is provided by the top ten mining pools\footnote{\url{https://btc.com/stats/pool}}. Consequently, mining pools are probably the most important actors in the blockchain ecosystem.

Typically, a mining pool is maintained and coordinated by a \textit{pool manager}. The success of a pool depends on its computational power, therefore the miners must commit their resources to find a full solution through mining.  To estimate the computing power of each miner,  the pool manager sets an easier puzzle to solve and requests the solutions to this puzzle, which are called \emph{partial solutions}\footnote{This is done by setting an internal pool target value higher (easier) than the network target value for the original puzzle (see footnote 6). Note that the set of partial solutions for this easier target value is a superset of the full solutions. Therefore it is a very good approximation of the computational commitment of a miner.}. Miners submit these partial solutions to the manager, each of which forms a \emph{share}. In case one of these shares is a full solution, the pool manager gets the reward and distributes it among the pool members based on their submitted shares and a predefined \emph{reward sharing scheme}.  The choice of reward sharing scheme is a crucial design element in a mining pool. The scheme must ensure that the pool is economically viable and must provide miners with the right incentives to act in ways beneficial for the common good of the pool. 

\subsection{Our contribution}

In this paper, we analyze and design \emph{reward sharing schemes} in mining pools by proposing a comprehensive axiomatic perspective. We depart from existing literature in terms of framework in various ways. First, our axiomatic framework is not on the consensus protocols but on the mining pools in any of these protocols. Second, our model is not restricted to a static single block, since various schemes in practice pay the miners repetitively over time in various blocks. Third, we propose reward sharing schemes and  \emph{allocations} not on the miners in a pool but instead on the shares submitted by these miners. This is particularly enriching, because defining allocations on shares (rather than on miners), enables more granular and relevant parameters such as the submission time or order of the shares. Therefore, we can practically formulate any of the existing schemes under our unified axiomatic framework, e.g., the Slush and PPLNS\footnote{See \cite{rosenfeld2011analysis} for a comprehensive list of reward sharing schemes, and Section \ref{dis_section} for formal definitions.}, and allow designers to propose new ones. 


We formulate several desirable axioms for reward sharing schemes. In particular, we propose two axioms concerning fairness, i.e., how the awards for the shares should be redistributed\footnote{This redistribution could be either in absolute amount or in relative ratio. Fair reward sharing is crucial for pool stability which in turn affects overall system efficiency. Given the symmetric nature of the shares in a pool round, these fairness axioms represent a natural (stronger) versions of the celebrated ``population monotonicity'' axiom in the literature on fair allocation (e.g. \cite{thomson2016fair}).} when the round is delayed by an additional share. We show that, together with other axioms, each of these fairness axioms, \emph{absolute redistribution} and \emph{relative redistribution}, characterize two distinct classes of reward sharing schemes, the class of absolute fair and relative fair schemes. Thereafter, we characterize the generalized class of proportional reward schemes, i.e., $k$-pseudo proportional schemes, which satisfies both of these axioms simultaneously. Finally, by imposing an additional \emph{strict positivity} requirement, we single out the well-known proportional reward scheme. 
%
%

In a recent study by \cite{tovanich2021empirical} and \cite{belotti2018bitcoin}, it is shown that the reward sharing schemes
and pool fees influence miners’ decisions to join, change, or
exit from a mining pool. Therefore from the perspective of miners as well as the pool management, the choice of a reward sharing scheme is a vital element of a well functioning and stable mining pool. Hence, the miners  perception of fairness of the  reward sharing scheme  is crucial.

\subsection{Related Literature}
Our paper relates to the literature on miner's general incentives under the PoW protocol. The seminal paper by \cite{rosenfeld2011analysis} provides one of the earliest comprehensive mathematical analysis on the reward sharing schemes in mining pools. \cite{schrijvers2016incentive} discuss strategic behavior of miners under ``incentive compatibility''. \cite{zolotavkin2017incentive} investigate the interplay between incentive compatibility and the distribution of computational power among miners in the pool. \cite{lewenberg2015bitcoin} model mining pools from a cooperative game perspective while \cite{qin2018research} and \cite{chatzigiannis2019diversification} investigate competition among pools\footnote{For a list of further readings regarding miners' behaviour, see \cite{eyal2015miner,eyal2014majority,biais2019blockchain,babaioff2012bitcoin,sapirshtein2016optimal,fisch2017socially,carlsten2016instability,kiayias2016blockchain,cong2019decentralized}.}. 

In terms of methodology, our paper utilizes tools from the literature on economic design (see e.g., \cite{moulin1991axioms,young1994equity,thomson2018divide}) and is inspired by welfare economics, in particular the literature on distributional fairness (see e.g., \cite{roemer1996theories,moulin1987equal,young1994equity} and fair allocation in networks (see e.g., \cite{hougaard2018allocation,moulin2018fair}). As mentioned, we take an axiomatic approach to the analysis of reward sharing schemes: a list of desirable properties of generic schemes (axioms) are identified, and individual schemes are uniquely characterized by different sets of axioms. In turn, this allows a qualitative comparison of various reward sharing schemes  based on their axiomatic foundation.

There has been a steadily growing literature, pointing towards utilization of mechanism design in the context of blockchains. \cite{can2019economic} proposes the use of economic design on the consensus protocol, while \cite{hougaard2018optimal} provides a different rationale for PoW by decentralized socially optimal reward schemes. Concerning an axiomatic approach to PoW, \cite{leshno2019bitcoin} and \cite{chen2019axiomatic} propose characterizations (on the consensus protocols) with properties such as anonymity, collusion-proofness and sybil-proofness the last two of which are analogous to merging and splitting\footnote{Both papers propose characterizations of the proportional reward scheme which is well studied in the economic design literature and in bankruptcy models (see \cite{banker1981equity,moulin1987equal,chun1988proportional,de1999coalitional}).}.

The remaining part of the paper is organized as follows. In Section \ref{model_section} we propose our modeling framework and notation. Section \ref{axioms_section} defines and comments on various axioms including the two fairness conditions. Section \ref{fair_sharing_section} presents the characterization results. Finally, Section \ref{dis_section} closes with a discussion of well-known schemes under our proposed framework and logical independence of the characterizing axioms.

\section{Model and notation}\label{model_section}

Let $T$ denote the set of all possible \emph{time signatures}. Let $H$ denote the set of all possible \emph{hashes} and $\mathcal{S}\subsetneq H$ denote the set of all admissible hashes, i.e., hash values that satisfy the difficulty of the pool. A (partial solution) \emph{share}  is then a two-tuple $s=(h,t)\in \mathcal{S}\times T$ consisting of an admissible hash submitted to the pool at time\footnote{Since some known reward sharing schemes use time signature as input we denote the shares with $(h,t)$. We shall drop the time signature when it is redundant for the schemes we analyze.} $t$. We denote the time signature of a share $s$ by a function $\tau(s)$, i.e., for $s=(h,t)$, we have $\tau(s)=t$. The set of all ordered shares submitted in a pool is denoted by $S=\{s_1,s_2, s_3, \ldots, s_m\}\subsetneq\mathcal{S}\times T$ where the order is defined by the shares' time signatures. We do not associate shares with individual miners since all schemes that are used in practice are neutral towards the ``identity'' of the miner.

A {\it pool round} is an ordered set of shares ending with  a share, submitted by the pool, which is a full solution on the blockchain. So everytime a share submitted by the pool is a full solution on the blockchain, a pool round ends and a new round begins. We consider the partitioning of the submitted set of shares by (pool) rounds, and denote it by $\mathcal{P}(S)$: for instance, let $l=|\mathcal{P}(S)|$ and let $P_1, P_2, \ldots, P_l $ denote the set of shares submitted in rounds $1, 2, \ldots, l$ respectively. Let $H=(S, \mathcal{P}(S))$  denote the \textit{history} of the pool.

Given the partition of a history $\mathcal{P}(S)$ and any share $s\in S$, we denote the set of all shares in the same round as $s$ by $P(s)=X\in \mathcal{P}(S)$ such that $s\in X$. Given a round $P\in \mathcal{P}(S)$, we denote the relative rank of a share $s$ in $P$ by $\rho(s)=|\{s'\in P | \tau(s')<\tau(s)\}|+1$, i.e., the number of shares that are submitted in round $P$ up to, and including, $s$. We define the {\it length} of a round $P$ by the number of shares submitted in that round, i.e. $|P|$.

\begin{example}\label{exm1}
	Equation \ref{equation_history} below, illustrates an example of a history: shares in bold are full solutions.  
	\begin{align}\label{equation_history}
		S&=\{ s_1,\dots, \mathbf{s_{10}}, s_{11}\dots,\mathbf{s_{150}}, \ldots, s_{564},\dots,\mathbf{s_{589}} ,\ldots, \mathbf{s_m}\} \nonumber \\
		\mathcal{P}(S)&= \{ \{\underbrace{s_1,s_2,\dots, \mathbf{s_{10}}}_{P_1} \}, \{\underbrace{s_{11}\dots,\mathbf{s_{150}}}_{P_2}\}, \ldots, \{\underbrace{s_{564},\dots,\mathbf{s_{589}}}_{P_r}\}, \ldots, \{\underbrace{\dots, \mathbf{s_m}}_{P_l}\}\}
	\end{align}
	
	%
	
	Considering, for instance,  the share $s_{564}$ we see that it is the first share of round $r$, so $P(s_{564})=P_r=\{s_{564}, \ldots , s_{589}\}$ and $\rho(s_{564})=1$. The length of round $r$ is $|P_r|=26$.  
\end{example}

\begin{definition}\label{def_reward_rule}
	Given a history $H=(S, \mathcal{P}(S))$, we define a \emph{reward sharing scheme}, $\alpha,$ by awards $$\alpha(s,  H)\in \mathbb{R}_+$$ to every share $s\in S$.
\end{definition}

In practice, the payments to miners are made after a certain amount of time, therefore when considering the allocation of rewards, we analyze a round that is terminated and already confirmed on the Blockchain\footnote{The pool manager should wait for at least 6 block confirmation to make sure the block found by the pool ends up on the longest chain.}. For each pool round, the pool manager charges a fee  to compensate the costs of running the pool. For a history $H=(S,\mathcal{P}(S))$ we denote this fee by the mapping $f:  2^{S}\setminus \emptyset \rightarrow \mathbb{R}$. A generic fee for a round $P_r$ in that case would be $f(P_r)$. We denote the award to be distributed in this round by $R_r=B-f(P_r),$ where $B$ is the block reward assigned to the pool when finding a full solution. Next, we introduce one of the  most  intuitive schemes, also known as the \textit{Proportional scheme} below.
\begin{example}\label{exm2} The Proportional scheme is one of the most straightforward and intuitive reward sharing schemes. For any given round, the pool manager gets a fixed fraction of the reward and then it assigns to every share $s\in P(s),$ a proportion of the reward relative to the length of the corresponding round $|P(s)|$. Formally, for any history $H=(S,\mathcal{P}(S))$ and any share $s\in S$:
	$$\alpha(s , H)=\frac{R_r}{|P(s)|}.$$
\end{example}

In Section~\ref{formal_rules} we provide a brief discussion of various reward sharing schemes that are used in practice or proposed to improve existing practices.


For our axiomatic analysis below, we are inspired by the framework in  \cite{schrijvers2016incentive}. In particular,  we consider situations in which full solutions are delayed, i.e., when the pool round ends with the submission of additional share(s) at the end. In the simplest of such cases, we consider histories where a pool round is extended by one additional share at the end of the round. We therefore need some additional definitions. 

Formally, let $H=(S, \mathcal{P}(S))$ be the history of a pool. Let $s^*\not \in S$ be the so-called additional last share (into the $r^{th}$ round), i.e., $\tau(\underline{s})<\tau(s^*)<\tau(\overline{s})$ for all $\underline{s}\in P_r$ and for all $\overline{s}\in P_{r+1}$. The history $H'=(S', \mathcal{P}'(S'))$ is an \emph{extension of $H$ at the $r^{th}$ round} whenever $S'=S\cup \{s^*\}$ and:
\begin{itemize}
	\item $P'_k=P_k$  for all $k\neq r$,
	\item $P'_r =P_r \cup \{s^*\}$.
\end{itemize}

\begin{remark}\label{remark_stronger}
	Note that, one can consider the extension of a round at any position and not only at the end of the round. Similarly we could also consider extension of a round by appending multiple shares at the end. However both of these scenarios lead to very restrictive axioms in the upcoming section (see Remark~\ref{different_extensions}).
\end{remark}

Next,  we define the restriction of a history to a single round. Let $H=(S,\mathcal{P}(S))$ be any history. The \emph{restriction of $H$ to the $r^{th}$ round} is denoted as $H|_r=(P_r, \{P_r\})$. That is, the set of shares in the history only consists of those at the $r^{th}$ round and the only partition of the history is $P_r$.

\begin{example}\label{extension} Recall the situation in Example~\ref{exm1}. Below is the example of an extension of the history in Equation~\ref{equation_history} at the $r^{th}$ round by the full solution (share) $s^*$.
	\begin{align*}
		S&=\{ s_1,\dots, \mathbf{s_{10}}, s_{11}\dots,\mathbf{s_{150}}, \ldots, s_{564},\dots,s_{589} , \mathbf{s^*}, \ldots, \mathbf{s_m}\}\\
		\mathcal{P}(S)&= \{ \{\underbrace{s_1,s_2,\dots, \mathbf{s_{10}}}_{P_1} \}, \{\underbrace{s_{11}\dots,\mathbf{s_{150}}}_{P_2}\}, \ldots, \{\underbrace{s_{564},\dots,s_{589},\mathbf{s^*}}_{P_r}\}, \ldots, \{\underbrace{\dots, \mathbf{s_m}}_{P_l}\}\}
	\end{align*}
	
	Moreover, an example of restriction of the history in Equation~\ref{equation_history} to the $r^{th}$ round is
	$$H|_r= \Big( \{s_{564},\dots,\mathbf{s_{589}}\}, \, \Big\{\{\underbrace{s_{564},\dots,\mathbf{s_{589}}}_{P_r}\}\Big\}\,\ \Big)$$
\end{example}
%

\section{Axioms}\label{axioms_section}
In the following section we discuss five desirable properties of generic reward sharing schemes. Two of these reflect different aspects of fairness related to how the rewards must be re-distributed in case there is one additional share in the round\footnote{In this sense our fairness axioms can be viewed as relational axioms in case of a variable populations framework, see e.g., \cite{thomson2016fair}. }.

The first property ensures a \textit{fixed total reward} to the miners for any pool round in a history. That implies that the fee charged by the pool manager is the same for any two rounds in a history. This guarantees that  the pool manager can not take advantage (or be harmed) from shorter (longer) rounds, as the pool must distribute the same amount to the miners in any round. This can be seen as a desirable feature, especially because it reduces volatility and uncertainty of miners' income, which is the main driving force behind forming pools. A miner could potentially be more interested in working under well-defined income schemes and avoid arbitrary fee chargers by the pool manager. Formally,

$\bullet$ \textbf{Fixed Total Reward:} A scheme $\alpha$ satisfies fixed total reward whenever, for any history $H=(S, \mathcal{P}(S)),$ and any two rounds $P, P'\in \mathcal{P}(S)$, we have $$\sum\limits_{s\in P} \alpha(s , H) = \sum\limits_{s\in P'} \alpha(s , H).$$ 

%
%
%

Next, we consider situations where the submission time of shares may change. To analyze such situations, we create a ceteris paribus case where the time signature of only a single share in a history changes in a ``minimal'' way. Note that in pools, shares are typically submitted very frequently hence by a minimal delay we mean situations which essentially do not affect the order in which these shares are submitted, e.g., negligible network lags in milliseconds. Formally, let $H=(S, \mathcal{P}(S))$ be any history,  and consider any round $P_r\in \mathcal{P}(S)$, and any share $s_i\in P_r$ in this round. A \emph{time-shift} $H'=(S', \mathcal{P}'(S'))$ of $H$ at $s_i$ is defined as:
\vspace{-0.2cm}
\begin{enumerate}
	\item $P'_j=P_j$ for all $j\neq r$, and
	\item $S'=(S\setminus \{s_i\})\cup \{s'_i\}$ for some $s'_i \in \mathcal{S}\setminus S$ such that $\tau(s_{i-1})<\tau(s'_i) < \tau(s_{i+1})$,
\end{enumerate}

Point $1$ above simply says in both histories all rounds are identical except round $r$. Point $2$ says in round $r$ everything is the same except share $s_i$ which moved tiny bit later in history $H'$ but still between the same share as in history $H$. The next condition, dubbed \textit{ordinality}, requires that time-shifts should not affect the reward distribution, so long as the order of shares is preserved.

$\bullet$ \textbf{Ordinality:} A scheme $\alpha$ satisfies ordinality if, for any $H=(S,\mathcal{P}(S))$, and for any time-shift $H'=(S', \mathcal{P}'(S'))$ of $H$ at any $s_i$, we have:
$$\alpha(s , H) = \alpha(s, H') \text{ for all } s\in S\setminus \{s_i\}$$

The next condition, dubbed \textit{budget limit},  requires that the pool manager charges a nonnegative fee. This ensures that the pool does not go bankrupt.

$\bullet$ \textbf{Budget Limit:} A scheme $\alpha$ satisfies budget limit whenever, for any history $H=(S, \mathcal{P}(S)),$ and any round $P\in \cal{P}(S)$, we have:
$$\sum\limits_{s\in P} \alpha(s , H) \leq B.$$

We now turn to our two main fairness conditions. The first is dubbed \textit{absolute redistribution}, and requires that in case a round is extended (delayed) by one additional share, the award assigned to any existing share in the corresponding round decreases by the same \emph{amount}\footnote{This is, in fact, a strong version of Population monotonicity in \cite{thomson2016fair} and resembles the spirit of Myerson's fairness axiom \cite{myerson1977graphs} stating that agents should be affected equally from entering mutual agreements.}. Formally,

$\bullet$ \textbf{Absolute redistribution:} A scheme $\alpha$ satisfies absolute redistribution 
whenever, for any history $H=(S,\mathcal{P}(S))$, any round $P_r$ with $|P_r|>1$, and any extension $H'=(S', \mathcal{P}'(S'))$ at the $r^{th}$ round we have   for any   $s_i,s_j\in P_r$:
$$\alpha(s_i , H) - \alpha(s_i , H')=\alpha(s_j , H) - \alpha(s_j , H').$$

In the same spirit, the next condition, dubbed \textit{relative redistribution}, requires that in case a round is extended (delayed) by one additional share, the reward of each existing share in the corresponding round is decreased by the same \emph{ratio}. Formally,

$\bullet$ \textbf{Relative redistribution:} A scheme $\alpha$ satisfies relative redistribution 
whenever for any history $H=(S,\mathcal{P}(S))$, any round $P_r$ with $|P_r|>1$, and any extension $H'=(S', \mathcal{P}'(S'))$ at the $r^{th}$ round, we have  for any  $s_i,s_j\in P_r$  with  $\alpha(s_i , H)\neq 0$  and  $\alpha(s_j , H)\neq 0$:
$$\frac{ \alpha(s_i , H')}{\alpha(s_i , H)}=\frac{ \alpha(s_j , H')}{\alpha(s_j , H)}.$$

\begin{remark}\label{different_extensions}
	As hinted in Remark~\ref{remark_stronger}, the definition of extension is crucial for both of the above axioms. In case we
	construct the extensions with not appending a single share at the end but instead arbitrarily positioning the share in the round, we would be strengthening both fairness axioms. Similarly, in case we construct the extensions with multiple shares, again these two fairness axioms will be strengthened, hence leading to a restrictive framework possibly leading to impossibilities.
\end{remark}

The final condition, dubbed \textit{round based rewards}, requires that the distribution of the reward only depends on the round itself and it is not affected by any other rounds in the history.  Formally,

$\bullet$ \textbf{Round based rewards:} A scheme $\alpha$ satisfies round based rewards whenever, for any history $H$, and any round $P_r$, we have for all  $s\in P_r$: $$\alpha(s, H) = \alpha(s , H|_r).$$ 

We first observe that round based rewards strengthens the fixed total reward condition, such that the latter imposes the fixed total rewards for any rounds in any history. 

\begin{lemma}\label{lemma1}
	If a scheme $\alpha$ satisfies fixed total reward and round based rewards, then for any two histories $H=(S, \mathcal{P}(S))$ and $H'=(S', \mathcal{P}'(S'))$ and any two rounds $P\in \mathcal{P}(S)$ and $P'\in \mathcal{P}'(S')$ we have
	$$\sum\limits_{s\in P} \alpha(s , H) = \sum\limits_{s\in P'} \alpha(s , H').$$
\end{lemma}
\begin{proof}
	See Appendix~\ref{lemma1_proof}.
\end{proof}

\begin{remark}\label{remark1}
	An immediate consequence of Lemma~\ref{lemma1} is that, if a scheme satisfies fixed total rewards and round based rewards then the fee must be the same for all histories and all rounds in these histories.
	Therefore, these two axioms, together with budget limit imply a stronger version of the budget limit\footnote{This is conceptually similar to the \textit{strong budget balance} in \cite{chen2019axiomatic}.} in the sense that for a scheme which satisfies  fixed total rewards and round based rewards and for any history $H=(S, \mathcal{P}(S))$ and any round $P\in \cal{P}(S)$, we have $\sum\limits_{s\in P} \alpha(s , H) =R$, with $R=B-f$ for some $f\in[0,B]$.
\end{remark}

Second, we observe that joint with ordinality, fixed total reward and round based rewards imply that shares with same relative rank in rounds of equal length get identical awards. Formally, 

\begin{lemma}\label{lemma2}
	If a scheme $\alpha$ satisfies fixed total reward, round based rewards and ordinality, then for any two histories $H=(S, \mathcal{P}(S))$ and $\bar{H}=(\bar{S}, \mathcal{\bar{P}}(\bar{S}))$ and any two rounds $P\in \mathcal{P}(S)$ and $\bar{P}\in \mathcal{\bar{P}}(\bar{S})$ such that $|P|=|\bar{P}|$, we have for all $s\in P$ and for all $\bar{s}\in \bar{P}$ such that $\rho(s)=\bar{\rho}(\bar{s})$:
	$$\alpha(s , H) = \alpha(\bar{s} , \bar{H}).$$
\end{lemma}
\begin{proof}
	See Appendix~\ref{lemma2_proof}.
\end{proof}

\section{Fair reward sharing schemes}\label{fair_sharing_section}
\subsection{Absolute Fairness}

In this section we single out a particular class of schemes, called  the class of {\it absolute fair schemes}. We show that this is the only such class that  satisfies absolute redistribution together with fixed total reward, budget limit, round based rewards, and ordinality.

\noindent
Specifically, a scheme belongs to this class if there exists  $\varepsilon: \mathbb{N}\rightarrow [0,1]$ with $\varepsilon(1)=1$ and  $\varepsilon(\rho(s))\geq \sum\limits_{i=\rho(s)+1}^{\infty}\frac{\varepsilon(i)}{i-1}$ such that,
\begin{equation} \label{afs}
	\alpha^{**}(s, H) = R\bigg(\varepsilon(\rho(s))- \sum\limits_{i=\rho(s)+1}^{|P(s)|}\frac{\varepsilon(i)}{i-1}\bigg)
\end{equation}

Note that for the last share $\rho(s)=|P(s)|$, therefore $\sum\limits_{i=|P(s)|+1}^{|P(s)|}\frac{\varepsilon(i)}{i-1}$ is an empty sum and by convention it equals $0$.

We can interpret absolute fair schemes in an iterative fashion: if there is only one share $s$, then  $\rho(s) = 1$ so $\varepsilon(1) = 1$ and the share receives the full net reward $R$. If we add a share $s'$, then $\rho(s') = 2$ and $\varepsilon(2) < 1$ so the first share $s$ gets $R(\varepsilon(1)-\varepsilon(2))$ and the second share $s'$ gets $R\varepsilon(2).$ Now, adding a third share $s''$ the award to both the first and second share should be reduced by the same amount $\delta$. Thus, $\varepsilon(3) = 2\delta$ and consequently the first share is awarded $R(\varepsilon(1) - \varepsilon(2) - \frac{\varepsilon(3)}{2})$ and the second share is awarded $R(\varepsilon(2) - \frac{\varepsilon(3)}{2}),$ and so forth, for any share in the round.


Next, we present our first main result. 

\begin{theorem}\label{absolute_thm}
	A reward sharing scheme $\alpha$ satisfies round based rewards, budget limit,  fixed total reward, ordinality and absolute redistribution  if and only if it is an absolute fair scheme in the sense  of~(\ref{afs}).
\end{theorem}
\begin{proof}
	See Appendix \ref{absolute_thm_proof}.
\end{proof}

A prominent member of the class of absolute fair schemes is the proportionals scheme as we shall demonstrate below. 
\begin{prop}\label{absolute_proportional}
	The proportional scheme is an absolute fair scheme.
\end{prop}
\begin{proof}
	Let $H=(S,\mathcal{P}(S))$ be any history.  It is easy to see the proportional scheme satisfies the round-baseness condition, hence we only consider a history with a single round, i.e., $H=(\big\{s_1,\dots,s_k\big\}, \big\{\{s_1,\dots,s_k\}\big\})$. Let  $\varepsilon(\rho(s_i))=\frac{1}{\rho(s_i)}$ for all $i\in\{1,\dots,k\}$. Note that $\rho(s_i)=i$. Therefore, for any share $s_j$ we have:
	\begin{align*}
		\alpha^{**}(s_j, H)&=R\bigg(\varepsilon(\rho(s_j))- \sum\limits_{i=\rho(s_j)+1}^{|P(s_j)|}\frac{\varepsilon(i)}{i-1}\bigg)=R\bigg(\frac{1}{j}- \sum\limits_{i=j+1}^{k}\frac{1}{i\times (i-1)}\bigg)\\
		&=R\bigg(\frac{1}{j}- \sum\limits_{i=j+1}^{k}\Big(\frac{1}{ i-1}-\frac{1}{i}\Big)\bigg)=R\bigg(\frac{1}{j}- \Big(\frac{1}{j}-\frac{1}{k}\Big)\bigg)=\frac{R}{k}
	\end{align*}
	Note that for the last share $\sum\limits_{i=k+1}^{k}\varepsilon(i)$ is an empty sum and by convention it equals $0$.
\end{proof}

\subsection{Relative Fairness}
In this section we single out a particular class of schemes, called  the class of {\it relative fair schemes}. We show that this is the only such class that  satisfies relative redistribution together with fixed total reward, budget limit, round based rewards, and ordinality.

\noindent
Specifically, a scheme belongs to this class if there exists  $\varepsilon(j): \mathbb{N}\rightarrow [0,1]$ with $\varepsilon(1)=1$ such that,
\begin{equation}\label{rfs}
	\alpha^*(s, H) =
	R \varepsilon(\rho(s))  \prod\limits_{j=\rho(s)+1}^{|P(s)|}\big(1-\varepsilon(j)\big)
\end{equation}

Note that for the last share $\rho(s)=|P(s)|$, therefore $\prod\limits_{j=|P(s)|+1}^{|P(s)|}(1-\varepsilon(j))$ is an empty product and by convention it equals $1$.

An iterative interpretation of relative fair schemes is as follows: if there is only one share $s$, $\varepsilon(1)=1$ and the share is awarded the entire net reward $R$. If we add a second share $s'$, then $s'$ will be awarded $R\varepsilon(2)$ and the  first share must get $R(\varepsilon(1)(1-\varepsilon(2))) = R(1- \varepsilon(2)).$ Now, adding  third share $s''$ we must now have that both the award to the first share and the second share change by the same ratio $\delta$ so $(1-\varepsilon(2))\delta + \varepsilon(2)\delta + \varepsilon(3) = 1$ implying that $\delta = 1- \varepsilon(3)$ yielding the award $R((1-\varepsilon(2))(1-\varepsilon(3)))$ to the first share and $R(\varepsilon(2)(1-\varepsilon(3)))$ to the second share,  and finally $R\varepsilon(3)$ to the third share, and so forth for any share in the round. 

We now present our second main result characterizing the class of relative fair reward sharing schemes.

\begin{theorem}\label{relative_thm}
	A reward scheme $\alpha$ satisfies round based rewards, budget limit,  fixed total reward, ordinality and relative redistribution  if and only if it is a relative fair scheme in the sense  of (\ref{rfs}).
\end{theorem}
\begin{proof}
	See Appendix \ref{relative_thm_proof}.
\end{proof}

It turns out that the proportional scheme also takes up a prominent position among the relative fair schemes as recorded below.
\begin{prop}\label{relative_proportional}
	The proportional scheme is a relative fair scheme.
\end{prop}
\begin{proof}
	Let $H=(S,\mathcal{P}(S))$ be any history.  It is easy to see the proportional scheme satisfies the round-baseness condition, hence we only consider a history with a single round, i.e., $H=(\big\{s_1,\dots,s_k\big\}, \big\{\{s_1,\dots,s_k\}\big\})$. Let  $\varepsilon(\rho(s_i))=\frac{1}{\rho(s_i)}$ for all $i\in\{1,\dots,k\}$. Note that $\rho(s_i)=i$. Therefore, for any share $s_j$ we have:
	\begin{align*}
		\alpha^*(s_j, H) &= R\times \varepsilon(\rho(s_j)) \times\prod_{i=\rho(s_j)+1}^{k}(1-\varepsilon(i))\\
		& =R\times(\frac{1}{j})\times(1-\frac{1}{j+1})\times(1-\frac{1}{j+2})\times\dots\times(1-\frac{1}{k})\\
		&=R\times (\frac{1}{j})\times(\frac{j}{j+1})\times(\frac{j+1}{j+2})\times\dots\times (\frac{k-2}{k-1})\times (\frac{k-1}{k})=\frac{R}{k}
	\end{align*}
	Note that for the last share $\prod\limits_{i=k+1}^{k}(1-\varepsilon(i))$ is \textit{empty product} and by convention equals $1$.
\end{proof}

\subsection{Consensus between  absolute and relative redistribution}\label{intersection_section}

Next, we search for a possible consensus between the two fairness concepts that we discussed in the previous sections. We find that the two aforementioned classes of absolute and relative fair schemes are not mutually exclusive. Therefore, imposing both fairness axioms together with the others conditions characterize a new class of schemes at the intersection of the former two classes. We call this new class of schemes \emph{$k$-pseudo proportional schemes} and find that this class is also a generalization of the well-known proportional reward scheme mentioned in Example~\ref{exm2}. Formally, given any $k$ and $\delta$, a $k$-pseudo proportional reward scheme assigns awards to shares in a round identical to the proportional scheme, $\frac{R}{|P(S)|}$, so long as the round is shorter than $k$, i.e., $|P(S)|<k$. In case the round has more shares than $k$, then the scheme assigns an award of i) $\delta$ to the $k^{th}$ share, ii) distributes the rest, $R-\delta$, to the first $k-1$ shares, and iii) awards $0$ to any share that is submitted after $k$. The general structure of the $k$-pseudo proportional scheme with $k>1$ is as follows:

%

\begin{equation}\label{pseudo_prop_def}
	\alpha^{k,\delta}(s,H) = \left\{\begin{array}{ll}
		\frac{R}{|P(s)|}, & \text{ if } |P(s)|<k\\
		\frac{R-\delta}{k-1}, & \text{ if } |P(s)|\geq k\text{ and }\rho(s)< k\\
		\delta, & \text{ if } |P(s)|\geq k \text{ and }\rho(s)=k\\
		0, & \text{ if } |P(s)|\geq k \text{ and } \rho(s)>k
	\end{array}\right.
\end{equation}
for $0\leq \delta \leq R$.


The following theorem shows that the class of $k$-pseudo proportional schemes is the only one at the intersection of relative fair schemes and absolute fair schemes.
\begin{theorem}\label{thm3}
	A reward sharing scheme satisfies round based rewards, budget limit,  fixed total reward, ordinality, absolute redistribution, and relative redistribution if and only if it is $k$-pseudo proportional in the sense of (\ref{pseudo_prop_def}). 
\end{theorem}
\begin{proof}
	See Appendix~\ref{thm3_proof}.
\end{proof}

Next, we consider a property that ensures strictly positive awards for all shares, for any history. To ensure that all shares get paid can be seen as a fairness requirement since all shares involve "work", but it also relates to miners' incentives. If miners are not paid for their shares they may stop mining for the pool: in particular, if they are not paid after a certain number of shares in a round as in the $k$-pseudo-proportional schemes.\footnote{A weaker version of the axiom would require only strict positivity of the last share in a round. This would likely be enough to incentivize miners to keep mining for the pool, e.g., as in the Pay-Per-Last-N-Shares scheme.}
We write this feature as an additional axiom and provide the resulting characterization of the proportional scheme as a corollary to Theorem~\ref{thm3}, in effect highlighting that the proportional scheme is the only member of the family (\ref{pseudo_prop_def}) that is compatible with miners' incentives to keep mining for the pool.


$\bullet$ \textbf{Strict positivity:}  A scheme $\alpha$ satisfies strict positivity whenever, for any history $H=(S, \mathcal{P}(S)),$ and any round $P\in \cal{P}(S)$, we have:
$$\alpha(s,H)\in \mathbb{R}_{++}.$$

\begin{corollary}
	The proportional rule is the only rule that satisfies round based rewards, budget limit, fixed total reward, ordinality, absolute redistribution, relative redistribution, and strict positivity.
\end{corollary}

\begin{proof} By Theorem~\ref{thm3} any rule that satisfies the first six axioms, is a $k$-pseudo proportional rule. As shown in the proof of Theorem~\ref{thm3}, in case $k=\infty$ this completes the proof.    Suppose for a contradiction, $k$ is finite. Consider any history with a round $P_r$ such that $|P_r|>k$. Obviously $\alpha^{k,\delta}(s_{k+1})=0$ which contradicts strict positivity. 
\end{proof}

\section{Conclusion}\label{dis_section}
This paper provides, for the first time, a rich framework for reward sharing schemes in mining pools through an economic design perspective. To demonstrate the flexibility in the design, we proposed various desirable axioms and put particular emphasis on fairness concepts. We provided three different characterizations of classes within this framework, i.e., absolute fair, relative fair, and $k$-pseudo proportional schemes.

In Appendix \ref{formal_rules}, we show that the framework also allows the formalization of various applied schemes and investigate these schemes axiomatically. The results are summarized in Table~\ref{tab2}. We also show that the provided axioms are logically independent in Appendix \ref{independence_section} and provide a summary in Table~\ref{tab1}.



\newpage
\bibliographystyle{chicago}
\bibliography{sample}

\newpage
\section*{Appendix}
\appendix
\section{Proof of Lemma \ref{lemma1}}\label{lemma1_proof}
\textbf{Lemma 1.}
If a scheme $\alpha$ satisfies fixed total reward and round based rewards, then for any two histories $H=(S, \mathcal{P}(S))$ and $H'=(S', \mathcal{P}'(S'))$ and any two rounds $P\in \mathcal{P}(S)$ and $P'\in \mathcal{P}'(S')$ we have
$$\sum\limits_{s\in P} \alpha(s , H) = \sum\limits_{s\in P'} \alpha(s , H').$$
\begin{proof}
	Take	 any two histories $H, H'$ and $P,P'$ as in the lemma. Let $r$ (and $r'$) denote the round number of $P$ (and $P'$) in history $H$ (and $H'$). As $\alpha$ is fixed total reward, there exist fixed rewards for both histories, say $K$ and $K'$: 
	\begin{equation}\label{lemma_1_eq_1}
		\sum\limits_{s\in P} \alpha(s , H)=K\text{ and }\sum\limits_{s\in P'} \alpha(s , H')=K'
	\end{equation}
	
	Consider the restriction of $H$ to $r^{th}$ round (and of $H'$ to $r'^{th}$ round). As $\alpha$ satisfies round based rewards, for all $s\in P$ (and for all $s'\in P'$) we have: $\alpha(s, H) = \alpha(s , H|_r) \text{ and } \alpha(s', H') = \alpha(s' , H'|_{r'})$. This implies:
	
	\begin{equation}\label{lemma_1_eq_2}
		\sum\limits_{s\in P} \alpha(s , H|_r) = K \text{ and } \sum\limits_{s\in P'} \alpha(s , H'|_{r'})=K'
	\end{equation}

	Next we consider two cases. We say two rounds in different histories \emph{overlap} when there are some shares in the these two rounds which have overlapping time signatures. Formally, we say $P$ and $P'$ are overlapping whenever: $\min\limits_{s\in P'}\tau(s)\leq\min\limits_{s\in P}\tau(s)\leq \max\limits_{s\in P'}\tau(s)$ or  $\min\limits_{s\in P'}\tau(s)\leq\max\limits_{s\in P}\tau(s)\leq \max\limits_{s\in P'}\tau(s)$. Based on whether these rounds $P$ and $P'$ in histories $H$ and $H'$ overlap or not, we shall separately prove the equality of rewards, $K=K'$, in Equation~\ref{lemma_1_eq_1}. 	
	
	\textbf{Case 1}. No overlap. Without loss of generality, assume all the shares in $P$ have an earlier time signature than those in $P'$. In this case, consider a history with two rounds $H''=(P\cup P', \{P, P'\})$. Note that $H''$ is a history consisting only of two rounds, $P$ as the first and $P'$ as the second round. Note also that $\max\limits_{s\in P}\tau(s) < \min\limits_{s\in P'}\tau(s)$.  
	
	As $\alpha$ is fixed total reward, there exist a fixed reward for this history, say $K''$:
	\begin{equation}\label{lemma_1_eq_1_case_1}
		\sum\limits_{s\in P} \alpha(s , H'')=\sum\limits_{s\in P'} \alpha(s , H'')=K''
	\end{equation}
	
	Consider the restriction of $H''$ to the $1^{st}$ and the $2^{nd}$ rounds.  As $\alpha$ satisfies round based rewards for all $s\in P$ we have $\alpha(s, H'') = \alpha(s , H''|_1)$, and for all  $s\in P'$ we have $\alpha(s, H'') = \alpha(s , H''|_2)$. This implies:
	
	\begin{equation}\label{lemma_1_eq_2_case_1}
		\sum\limits_{s\in P} \alpha(s , H''|_1)=K'' \text{ and } \sum\limits_{s\in P'} \alpha(s , H''|_2)=K''.
	\end{equation}
	
	Note that the restriction of $H''$ to the $1^{st}$ round is equivalent to the restriction of $H$ to the $r^{th}$ round. Similarly the restriction of $H''$ to the $2^{nd}$ round is equivalent to the restriction of $H'$ to the $r'^{th}$ round. That is, $H''|_1=(P, \{P\})= H|_r$ and $H''|_2= (P',\{P'\})=H'|_{r'}$, therefore we have:
	
	$$\sum\limits_{s\in P} \alpha(s , H''|_1)=\sum\limits_{s\in P} \alpha(s , H|_r) \text{ and } \sum\limits_{s\in P'} \alpha(s , H''|_2)=\sum\limits_{s\in P'} \alpha(s , H'|_{r'})$$
	
	The former equality implies $K''=K$ and the latter implies $K''=K'$. This completes the proof for Case 1.

	\textbf{Case 2}. Overlap. In this case we consider two additional histories with two rounds each: $\overline{H}=(P\cup \overline{P}, \{P, \overline{P}\})$ and $\overline{H}'=(P'\cup\overline{P}', \{P', \overline{P}'\})$ such that 
	$$\max\limits_{s\in P\cup P'}\tau(s)<\min\limits_{s\in  \overline{P}}\tau(s) < \max\limits_{s\in  \overline{P}}\tau(s)< \min\limits_{s\in  \overline{P}'}\tau(s).$$
	Remark that $\overline{P}$ and $\overline{P}'$ are constructed such that neither overlaps with $P$ or $P'$. Hence we can construct the two histories above, i.e., $\overline{H}$ and $\overline{H}'$. Note also that $\overline{H}$ is a history consisting only of two rounds, $P$ as the first and $\overline{P}$ as the second round. Similarly, $\overline{H}'$ is a history consisting only of two rounds, $P'$ as the first and $\overline{P}'$ as the second round. 
	
	%
	%
	%

	Since $\max\limits_{s\in P\cup P'}\tau(s)<\min\limits_{s\in  \overline{P}}\tau(s)$ there is no overlaps between $P$ and $\overline{P}$, therefore applying Case 1 on $H$ and $\overline{H}$ yields:
	\begin{equation}\label{lemma_1_firsthistory}
		i) \sum\limits_{s\in P} \alpha(s , H) = \sum\limits_{s\in\overline{P}} \alpha(s , \overline{H})  
	\end{equation}
	
	Since $\max\limits_{s\in P\cup P'}\tau(s)< \min\limits_{s\in  \overline{P}'}\tau(s)$ there is no overlaps between $P'$ and $\overline{P}'$, therefore applying Case 1 on $H'$ and $\overline{H}'$ yields:
	\begin{equation}\label{lemma_1_secondhistory}
		ii) \sum\limits_{s\in P'} \alpha(s , H') = \sum\limits_{s\in\overline{P}'} \alpha(s , \overline{H}').
	\end{equation}

	Since $\max\limits_{s\in  \overline{P}}\tau(s)< \min\limits_{s\in  \overline{P}'}\tau(s)$ there is no overlaps between $\overline{P}$ and $\overline{P}'$, therefore applying Case 1 on $\overline{H}$ and $\overline{H}'$  yields:
	\begin{equation*}\label{lemma_1_combining_the_two}
		\sum\limits_{s\in\overline{P}} \alpha(s , \overline{H}) = \sum\limits_{s\in\overline{P}'} \alpha(s , \overline{H}').
	\end{equation*}
	
	which -combined with Equations~\ref{lemma_1_firsthistory} and~\ref{lemma_1_secondhistory}- completes the proof. 
\end{proof}

\section{Proof of Lemma \ref{lemma2}}\label{lemma2_proof}
\textbf{Lemma 2.}
If a scheme $\alpha$ satisfies fixed total reward, round based rewards and ordinality, then for any two histories $H=(S, \mathcal{P}(S))$ and $\bar{H}=(\bar{S}, \mathcal{\bar{P}}(\bar{S}))$ and any two rounds $P\in \mathcal{P}(S)$ and $\bar{P}\in \mathcal{\bar{P}}(\bar{S})$ such that $|P|=|\bar{P}|$, we have for all $s\in P$ and for all $\bar{s}\in \bar{P}$ such that $\rho(s)=\bar{\rho}(\bar{s})$:
$$\alpha(s , H) = \alpha(\bar{s} , \bar{H}).$$
\begin{proof}
	Take any scheme $\alpha$ that satisfies these two axioms.
	Let  $H=(S, \mathcal{P}(S))$ and $\bar{H}=(\bar{S}, \mathcal{\bar{P}}(\bar{S}))$ be two histories and let $P\in \mathcal{P}(S)$ and $\bar{P}\in \mathcal{\bar{P}}(\bar{S})$ be any two rounds with $|P|=|\bar{P}|=k$ for some $k$. Let $P=P_r$ and $\bar{P}= \bar{P}_{\bar{r}}$, i.e., $P$ is the $r^{th}$ round in $H$ and $\bar{P}$ is the $\bar{r}^{th}$ round in $\bar{H}$. As $\alpha$ satisfies round based rewards we have i) $\alpha(s, H)=\alpha(s, H|_{r})$ for all $s\in S$ and ii) $\alpha(\bar{s}, \bar{H})=\alpha(\bar{s}, \bar{H}|_{\bar{r}})$ for all $\bar{s}\in \bar{S}$. Let $H|_r=(\big\{s_1,\dots,s_k\big\}, \big\{\{s_1,\dots,s_k\}\big\})$ and $\bar{H}|_{\bar{r}}=(\big\{\bar{s}_1,\dots,\bar{s}_k\big\}, \big\{\{\bar{s}_1,\dots,\bar{s}_k\}\big\})$ denote these histories. In what follows we show for all $i\leq k$, 
	$$\alpha(s_i , H) = \alpha(\bar{s}_i , \bar{H})$$

	Consider a time-shift $\hat{H}^1$ of $H$ at $s_1$ (first share) with $\hat{S}^1=(S\setminus \{s_1\})\cup \{\hat{s}_1\})$ such that $\tau(\hat{s}_1)=\min \{ \tau(s_1), \tau(\bar{s}_1) \}$. Similarly consider a time-shift $\hat{\bar{H}}^1$ of $\bar{H}$ at $\bar{s}_1$ (first share) with $\hat{\bar{S}}^1=(\bar{S}\setminus \{\bar{s}_1\})\cup\{\hat{\bar{s}} \}$ such that $\tau(\hat{\bar{s}}_1)=\min \{ \tau(s_1), \tau(\bar{s}_1) \}$. Then by ordinality of $\alpha$, we have for all $s\in S\setminus \{s_1\}$ and for all $\bar{s}\in \bar{S}\setminus \{\bar{s}_1\}$:
	\begin{align*}
		\alpha(s, H|_r)= \alpha(\hat{s}, \hat{H}^1)\\
		\alpha(\bar{s}, \bar{H}|_{\bar{r}})= \alpha(\hat{\bar{s}}, \hat{\bar{H}}^1)
	\end{align*}
	
	Note that by Remark~\ref{remark1}, the total reward is fixed both in $P$ and in $\bar{P}$, therefore we conclude that for all $s\in S$ and for all $\bar{s}\in \bar{S}$:
	\begin{align*}
		\alpha(s, H|_r)= \alpha(\hat{s}, \hat{H}^1)\\
		\alpha(\bar{s}, \bar{H}|_{\bar{r}})= \alpha(\hat{\bar{s}}, \hat{\bar{H}}^1)
	\end{align*}
	
	Continuing iteratively and letting  $\hat{H}^i$ be the time-shift of $H$ at the $i^{th}$ share and $\hat{\bar{H}}^i$ as the time-shift of $\bar{H}$ at the $i^{th}$ share, the same argument above yields for all $i\leq k$:
	\begin{align*}
		\alpha(s, H|_r)= \alpha(\hat{s}, \hat{H}^1)= \alpha(\hat{s}, \hat{H}^2)=\dots=\alpha(\hat{s}, \hat{H}^k)\\
		\alpha(\bar{s}, \bar{H}|_{\bar{r}})=\alpha(\hat{\bar{s}}, \hat{\bar{H}}^1)=\alpha(\hat{\bar{s}}, \hat{\bar{H}}^2)=\dots= \alpha(\hat{\bar{s}}, \hat{\bar{H}}^k)
	\end{align*}	
	Note that by construction, the time signature of all shares at $\hat{H}^k$ are the same with those at $\hat{\bar{H}}^k$, i.e., $\hat{H}^k=\hat{\bar{H}}^k$. Therefore, $\alpha(\hat{s}, \hat{H}^k)=\alpha(\hat{\bar{s}}, \hat{\bar{H}}^k)$. This -together with the two equations above- implies $\alpha(s, H|_r)=\alpha(\bar{s}, \bar{H}|_{\bar{r}})$.
\end{proof}

\section{Proof of Theorem \ref{absolute_thm}}\label{absolute_thm_proof}
\textbf{Theorem 1.}
A reward sharing scheme $\alpha$ satisfies round based rewards, budget limit,  fixed total reward, ordinality and absolute redistribution  if and only if it is absolute fair in the sense of~\ref{afs}.
\begin{proof}
	\textbf{If part}. 
	\\ \textbf{Round based rewards}. Let $H=(S,\mathcal{P}(S))$ be a history, and let $r$ be any round. Note that, by restricting a round to the $r^{th}$ round the relative rank of a share $s$ in the round, as well as its time signature and value are the same at both $H$ and $H|_r$. Therefore,  $\alpha^{**}(s, H) = \alpha^{**}(s , H|_r)$.  
	
	\noindent
	\textbf{Budget limit}. Let $H=(S,\mathcal{P}(S))$ be any history. Take any round $r$. Let  $|P_r|=k$. As the relative fair scheme is based on $\rho(s)$, without loss of generality, we assume $P_r=\{1,2,\dots,k\}$ so that $\rho(s)=s$. Note that $	\sum\limits_{s=1}^{k} \alpha^{**}(s, H)=R\sum\limits_{s=1}^{k} \bigg(\varepsilon(s)- \sum\limits_{i=s+1}^{k}\frac{\varepsilon(i)}{i-1}\bigg)$. Therefore,
	
	\begin{align*}
		\frac{1}{R}\sum\limits_{s=1}^{k} \alpha^{**}(s, H)&=\sum\limits_{s=1}^{k} \bigg(\varepsilon(s)- \sum\limits_{i=s+1}^{k}\frac{\varepsilon(i)}{i-1}\bigg)\\
		&=	  \sum\limits_{s=1}^{k-1} \left(\varepsilon(s) - \sum\limits_{i=s+1}^{k}\frac{\varepsilon(i)}{i-1} \right)+ \varepsilon(k)  \\
		&= \sum\limits_{s=1}^{k-1}\varepsilon(s)  - \sum\limits_{s=1}^{k-1}\Big(\sum\limits_{i=s+1}^{k}\frac{\varepsilon(i)}{i-1}\Big) +\varepsilon(k) \\
		&= \sum\limits_{s=1}^{k-1}\varepsilon(s)  - \sum\limits_{s=1}^{k-1}
		\Big(\sum\limits_{i=s+1}^{k-1}\frac{\varepsilon(i)}{i-1}+\frac{\varepsilon(k)}{k-1}\Big) +\varepsilon(k) \allowdisplaybreaks\\
		&= \sum\limits_{s=1}^{k-1}\varepsilon(s)  - \sum\limits_{s=1}^{k-1}
		\sum\limits_{i=s+1}^{k-1}\frac{\varepsilon(i)}{i-1}-\sum\limits_{s=1}^{k-1}\frac{\varepsilon(k)}{k-1} +\varepsilon(k) \\
		&= \sum\limits_{s=1}^{k-1}\varepsilon(s)  - \sum\limits_{s=1}^{k-1}
		\sum\limits_{i=s+1}^{k-1}\frac{\varepsilon(i)}{i-1}-(k-1)\frac{\varepsilon(k)}{k-1} +\varepsilon(k) \\
		&= \sum\limits_{s=1}^{k-1}\varepsilon(s)  - \sum\limits_{s=1}^{k-1}
		\sum\limits_{i=s+1}^{k-1}\frac{\varepsilon(i)}{i-1} \\
		&\dots\dots\dots\\	
		&= \sum\limits_{s=1}^{2}\varepsilon(s)  - \sum\limits_{s=1}^{2}
		\sum\limits_{i=s+1}^{2}\frac{\varepsilon(i)}{i-1} \\
		&=(\varepsilon(1)+\varepsilon(2))-(\varepsilon(2))=\varepsilon(1)
	\end{align*}

	By definition of the $\varepsilon$  we have $0\leq\varepsilon(1)\leq 1$, hence $\sum\limits_{s=1}^{k} \alpha^{**}(s, H)\leq R$ the absolute redistribution scheme satisfies the budget limit.
	
	\noindent
	\textbf{Fixed total rewards}. It follows from a similar proof of the above.
	
	\textbf{Absolute redistribution}. Let $H=(S,\mathcal{P}(S))$ be any history and $P_r=\{s_1,s_2,\dots,s_k\}$ be any round.  Let $H'=(S', \mathcal{P}'(S'))$ be any extension  of $H$ at the $r^{th}$ round.  Take  any $s_a\in P_r$. Note that Note that $|P'(s_a)|=|P(s_a)|+1$.
	\begin{align*}
		\alpha^{**}(s_a , H) - \alpha^{**}(s_a , H')&=R\bigg(\varepsilon(\rho(s_a))- \sum\limits_{i=\rho(s_a)+1}^{|P(s_a)|}\frac{\varepsilon(i)}{i-1}\bigg)
		-
		R\bigg(\varepsilon(\rho(s_a))- \sum\limits_{i=\rho(s_a)+1}^{|P'(s_a)|}\frac{\varepsilon(i)}{i-1}\bigg)\\
		&=-R\bigg( \sum\limits_{i=\rho(s_a)+1}^{|P(s_a)|}\frac{\varepsilon(i)}{i-1}\bigg)
		+
		R\bigg( \sum\limits_{i=\rho(s_a)+1}^{|P(s_a)|+1}\frac{\varepsilon(i)}{i-1}\bigg)\\
		&=R\frac{\varepsilon(|P(s_a)|+1)}{|P(s_a)|}=R\frac{\varepsilon(k+1)}{k}
	\end{align*}

	\noindent\textbf{Only if part}. Take any reward sharing scheme $\alpha$ that satisfies the axioms. Take any history $H=(S,\mathcal{P}(S))$ and any $s\in S$. Let $s\in P_r$ for some $1\leq r\leq l=|\mathcal{P}(S)|$. As $\alpha$ satisfies round based rewards $\alpha(s, H)=\alpha(s, H|_r)$.  Similarly for $\alpha^{**}$, we have $\alpha^{**}(s, H)=\alpha^{**}(s, H|_r)$.  Hence, it suffices to prove $\alpha(s, H|_r)=\alpha^{**}(s, H|_r)$ at the restricted histories, i.e.,  histories with only a single round.

	Note that by Lemma~\ref{lemma2}, if the rounds in any two restricted histories are of the same size, then the shares are awarded based on their ranks. Hence, without loss of generality, we can denote these rounds $P_r=\{s_1,s_2,\ldots, s_n\}$ and the restricted histories as $H^n=H|_r$. Therefore, it suffices to prove $\alpha(s_i, H^n)=\alpha^{**}(s_i, H^n)$ for all $n$ and for all $i\leq n$. 
	%
	%
	%
	%
	%
	%
	%
	%
	%
	%

	%

	By absolute redistribution, for any $j>1$, we can let $\delta_{j}$ denote the absolute decrease in the rewards of all shares whilst moving from history $H^{j-1}$ to $H^{j}$, i.e., $\alpha(s,H^{j-1})-\alpha(s,H^j)$. In addition let $\delta_1=\alpha(s_1,H^1)$. As $\alpha$ is well-defined, for all $n\geq 1$, $\alpha(s_1,H^n)\geq 0$. Therefore, $\delta_1\geq \sum\limits_{i=2}^{\infty}\delta_i$ (otherwise at some history, $s_1$ would get a negative reward). Similarly, as $\alpha$ is well-defined, for all $n>1$, and for all $p\leq n$, $\alpha(s_p,H^n)\geq 0$. Note that by construction and Remark~\ref{remark1}, $\alpha(s_p,H^p)=(p-1)\delta_p$. Therefore, $(p-1)\delta_p\geq \sum\limits_{i={p+1}}^{\infty}\delta_i$ (otherwise at some history, $s_p$ would get a negative reward). Now, consider a function $\varepsilon$ such that $\varepsilon(1)=1$ and for all $j>1$, 
	
	\begin{equation}\label{epsilon_function}
		\varepsilon(j)=\frac{(j-1)\delta_{j}}{R}
	\end{equation}

	Note that we have the following two properties for $\varepsilon$: 

	\begin{enumerate}
		\item for all $j\geq 1$, $\varepsilon(j)\rightarrow [0,1]$. By definition $\delta_{j}$ denotes the absolute decrease in the rewards of all shares whilst moving from history $H^{j-1}$ to $H^{j}$. Therefore, the last share at $H^{j}$ gets $(j-1)\delta_{j}$ (as there are $j-1$ that each payouts $\delta_{j}$). As the scheme satisfies fixed total rewards and round based rewards then by Remark~\ref{remark1} the total reward at round   $H^{j}$ is $R$. Therefore, it must be the case that $(j-1)\delta_{j}\leq R$ so $\delta_{j}\leq \frac{R}{j-1}$. Therefore, $\varepsilon(j)=\frac{(j-1)\delta_{j}}{R}\leq \frac{j-1}{R}\times \frac{R}{j-1}$. As $\delta_{j}\geq 0$ then $\varepsilon(j)\geq 0$. All in all,  $\varepsilon(j)\rightarrow [0,1]$ for all $j\geq 1$.
		\item for all $j\geq 1$, $\varepsilon(j)\geq \sum\limits_{i=j+1}^{\infty}\frac{\varepsilon(i)}{i-1}$.
		
		For $j=1$ we have $\varepsilon(j)=1$ we have $1\geq \sum\limits_{i=2}^{\infty}\frac{\varepsilon(i)}{i-1}$. Replacing from Equation~\ref{epsilon_function} we have $1\geq \sum\limits_{i=2}^{\infty}\frac{1}{i-1}\frac{(i-1)\delta_{i}}{R}$ which implies $R\geq \sum\limits_{i=2}^{\infty}\delta_{i}$.
		
		For $\varepsilon(j)$ we have $\varepsilon(j)=\frac{(j-1)\delta_{j}}{R} \geq \sum\limits_{i=j+1}^{\infty}\frac{\varepsilon(i)}{i-1}$. Replacing from Equation~\ref{epsilon_function} we have $(j-1)\delta_j\geq \sum\limits_{i=j+1}^{\infty}\delta_{i}$.
	\end{enumerate}

	In what follows we shall show that $\alpha(s,H^n)=\alpha^{**}(s,H^n)$ for the aforementioned $\varepsilon$ function. We do it first for a single share rounds and then for multi-share rounds.

	\textbf{Single-share round:} Let $H^1$ be any history with a single round with a single share $s$. By Remark 1, $\alpha(s, H^1)=R$. Setting $\varepsilon(1)=1$ we have:  
	\begin{equation}\label{singe_share_round_absolute}
		\alpha(s,H^1)=\alpha^{**}(s,H^1)=R\varepsilon(1)
	\end{equation}
	
	\textbf{Multi-share round:} Let $H^n$ be any history with a single round with multiple shares $s_1, \ldots, s_n$. In what follows, we show  that $\alpha(s,H^n)=\alpha^{**}(s,H^n)$ any $s\in P_r$, by induction on $n$, i.e., the size of the round $P_r$.

	\textbf{Induction Basis:} Let $n=2$. Let $H^2=(\{s_1,s_2\},\big\{\{s_1,s_2\}\big\} )$. We will show that  $\alpha(s_1, H^2)=\alpha^{**}(s_1, H^2)$ and $\alpha(s_2, H^2)=\alpha^{**}(s_2, H^2)$. By Remark \ref{remark1} for two histories $H^1$ and $H^2$, and by Equation~\ref{singe_share_round_absolute}, we have $\alpha(s_1, H^2)+\alpha(s_2, H^2)=R$. So, $\alpha(s_1, H^2)=R-\delta_2$ and $\alpha(s_2, H^2)=\delta_2$.. Setting $\varepsilon(2)=\frac{\delta_2}{R}$, yields $\alpha(s_1, H^2)=R\varepsilon(1)-R\varepsilon(2)=\alpha^{**}(s_1, H^2)$ and $\alpha(s_2, H^2)=R\varepsilon(2)=\alpha^{**}(s_2, H^2)$.

	\textbf{Induction Hypothesis:} Let $n=k$ with $k>1$. Suppose we have $\alpha(s_i, H^k)=\alpha^{**}(s_i, H^k)$ for all $i\leq k$.
	
	To prove for $n=k+1$, consider any $H^{k+1}=(\big\{s_1,\dots,s_k,s_{k+1}\big\}, \big\{\{s_1,\dots,s_k,s_{k+1}\}\big\})$. We will show that  $\alpha(s_i, H^{k+1})=\alpha^{**}(s_i, H^{k+1})$ for all  $i\leq k+1$. Let $H^k=(\big\{s_1,\dots,s_k\big\}, \big\{\{s_1,\dots,s_k\}\big\})$.

	By construction, $\alpha(s_i , H^{k+1})=\alpha(s_i , H^{k})-\delta_{k+1}$ for all $i\leq k$. By induction hypothesis 
	$\alpha(s_i, H^k)=\alpha^{**}(s_i, H^k)$ for all $i\leq k$ which implies $\alpha(s_i, H^{k+1})=\alpha^{**}(s_i, H^k)-\delta_{k+1}$. Therefore,
	\begin{equation}\label{equation_1_to_k_absolute_111}
		\alpha(s_i, H^{k+1})=R\bigg(\varepsilon(i)- \sum\limits_{j=i+1}^{k}\frac{\varepsilon(j)}{j-1}\bigg) - \delta_{k+1}\text{ for all } i\leq k
	\end{equation}
	
	As $\delta_{k+1}=\frac{\varepsilon(k+1)R}{k}$, the above equation is simplified as:
	
	\begin{equation}\label{equation_1_to_k_absolute}
		\alpha(s_i, H^{k+1})=R\bigg(\varepsilon(i)- \sum\limits_{j=i+1}^{k+1}\frac{\varepsilon(j)}{j-1}\bigg)\text{ for all } i\leq k
	\end{equation}

	\noindent Note that $\alpha(s_i , H^{k+1})=\alpha(s_i , H^{k})-\delta_{k+1}$, and by Lemma~\ref{lemma1} (on $H^k$ and on $H^1$), we have 
	\begin{equation}\label{absolute_equationsomething}
		\sum\limits_{i=1}^{k}\alpha(s_i , H^{k+1})=\sum\limits_{i=1}^{k}\alpha(s_i , H^{k})-k\delta_{k+1}=R-k\delta_{k+1}
	\end{equation}
	
	\noindent Note also that 
	$ \sum\limits_{i=1}^{k+1}\alpha(s_i , H^{k+1})=\sum\limits_{i=1}^{k}\alpha(s_i , H^{k+1})+\alpha(s_{k+1},H^{k+1})$. Plugging Equation~\ref{absolute_equationsomething} into this yields $\sum\limits_{i=1}^{k+1}\alpha(s_i , H^{k+1})=R-k\delta_{k+1}+\alpha(s_{k+1} , H^{k+1})$. By Remark~\ref{remark1}, $\sum\limits_{i=1}^{k+1}\alpha(s_i , H^{k+1})=R$, therefore we have:
	\begin{equation}\label{absolute_equation_k+1}
		\alpha(s_{k+1} , H^{k+1})= k\delta_{k+1}=R\varepsilon(k+1).
	\end{equation}
	
	Equations~\ref{equation_1_to_k_absolute} and~\ref{absolute_equation_k+1} together prove $\alpha(s_i, H^{k+1})=\alpha^*(s_i, H^{k+1})$ for $i\in \{1, \ldots, k\}$ and for $i=k+1$, respectively.
\end{proof}

\section{Proof of Theorem \ref{relative_thm}}\label{relative_thm_proof}

\textbf{Theorem 2.}
A reward scheme $\alpha$ satisfies round based rewards, budget limit,  fixed total reward, ordinality and relative redistribution  if and only if it is relative fair    in the sense  of \ref{rfs}.
\begin{proof}
	\textbf{If part}.\\
	\textbf{Round based rewards}.To show that the relative fair scheme satisfies $\alpha^*(s, H) = \alpha^*(s , H|_r)$, let $H=(S,\mathcal{P}(S))$ be a history, and let $r$ be any round. Note that, by restricting a history to the $r^{th}$ round, the round and the shares remain intact, hence the relative rank of a share $s$ is the same at both $H$ and $H|_r$. Therefore, the relative fair scheme satisfies round based rewards.  \\
	\textbf{Budget limit}. Let $H=(S,\mathcal{P}(S))$ be any history. Take any round $r$. Let  $|P_r|=k$. As the relative fair scheme is based on $\rho(s)$, without loss of generality, we assume $P_r=\{1,2,\dots,k\}$ so that $\rho(s)=s$. Therefore to  show $\sum\limits_{s\in P_r} \alpha^{**}(s , H) \leq R$, we have
	\begin{align*}
		&\sum\limits_{s\in P_r} \alpha(s, H) = 	\sum\limits_{s=1}^{k}\left(R\times \varepsilon(\rho(s)) \times\prod\limits_{i=\rho(s)+1}^{k}(1-\varepsilon(i))\right)\\
		&=	R \left( \varepsilon(1) \times\prod\limits_{i=2}^{k}(1-\varepsilon(i))\right) + R
		\sum\limits_{s=2}^{k}\left( \varepsilon(\rho(s)) \times\prod\limits_{i=\rho(s)+1}^{k}(1-\varepsilon(i))\right)\\
		&=	R \Big( (1-\varepsilon(2))(\prod\limits_{i=3}^{k}1-\varepsilon(i))\Big) + R
		\sum\limits_{s=2}^{k}\left( \varepsilon(s) \times\prod\limits_{i=s+1}^{k}(1-\varepsilon(i))\right)\\
		&=	R
		\Big(
		\prod\limits_{i=3}^{k} \Big(1-\varepsilon(i)\Big) 
		-
		\varepsilon(2) \prod\limits_{i=3}^{k} \Big(1-\varepsilon(i)\Big) 
		\Big) + R
		\sum\limits_{s=2}^{k}\left( \varepsilon(s) \times\prod\limits_{i=s+1}^{k}(1-\varepsilon(i))\right)\\
		&=	R
		\prod\limits_{i=3}^{k} \Big(1-\varepsilon(i)\Big) 
		-
		R \varepsilon(2) \prod\limits_{i=3}^{k} \Big(1-\varepsilon(i)\Big) 
		+ R \varepsilon(2) \prod\limits_{i=3}^{k}(1-\varepsilon(i))
		+
		R
		\sum\limits_{s=3}^{k}\left( \varepsilon(s) \times\prod\limits_{i=s+1}^{k}(1-\varepsilon(i))\right)
	\end{align*}
	\begin{align*}
		&=	R
		\prod\limits_{i=3}^{k} \Big(1-\varepsilon(i)\Big)  + R
		\sum\limits_{s=3}^{k}\left( \varepsilon(s) \times\prod\limits_{i=s+1}^{k}(1-\varepsilon(i))\right)\\
		&=\cdots\\
		&=	R\times \Big((1-\varepsilon(k)) \Big) + R
		\sum\limits_{s=k}^{k}\left( \varepsilon(s) \times\prod\limits_{i=s+1}^{k}(1-\varepsilon(i))\right)\\
		&=	R\times \Big((1-\varepsilon(k)) \Big) + R\varepsilon(k)=R
	\end{align*}
	
	\noindent \textbf{Fixed total rewards}. It follows from a similar proof of the above. 
	
	\noindent\textbf{Relative redistribution}. To show Relative fairness let $H=(S,\mathcal{P}(S))$ be any history and $P_r=\{s_1,s_2,\dots,s_k\}$ be any round. Let $H'=(S', \mathcal{P}'(S'))$ be any extension  of $H$ at the $r^{th}$ round. Then for any $s_a,s_b\in P_r$ $\alpha^*(s_a , H)\neq 0$ and $\alpha^*(s_b , H)\neq 0$ we have:
	
	\begin{align*}
		\frac{ \alpha^*(s_a , H')}{\alpha^*(s_a , H)}&= \frac{ R\times \varepsilon(\rho(s_a)) \times\prod\limits_{i=\rho(s_a)+1}^{|P(s_a)|+1}\big(1-\varepsilon(i)\big)}{R\times \varepsilon(\rho(s_a)) \times\prod\limits_{i=\rho(s_a)+1}^{|P(s_a)|}\big(1-\varepsilon(i)\big)}
		=1-\varepsilon\big(|P(s_a)|+1\big) =1-\varepsilon\big(k+1\big)
	\end{align*}
	\begin{align*}
		\frac{ \alpha^*(s_b , H')}{\alpha^*(s_b , H)}&= \frac{ R\times \varepsilon(\rho(s_b)) \times\prod\limits_{i=\rho(s_b)+1}^{|P(s_b)|+1}\big(1-\varepsilon(i)\big)}{R\times \varepsilon(\rho(s_b)) \times\prod\limits_{i=\rho(s_b)+1}^{|P(s_b)|}\big(1-\varepsilon(i)\big)}
		=1-\varepsilon\big(|P(s_b)|+1\big) =1-\varepsilon\big(k+1\big)
	\end{align*}
	The above two equations shows that $\frac{ \alpha^*(s_a , H')}{\alpha^*(s_a , H)}=\frac{ \alpha^*(s_b , H')}{\alpha^*(s_b , H)}$.

	\textbf{Only if part}. Take any reward sharing scheme $\alpha$ that satisfies the axioms. Take any history $H=(S,\mathcal{P}(S))$ and any $s\in S$. Let $s\in P_r$ for some $1\leq r\leq l=|\mathcal{P}(S)|$. As $\alpha$ satisfies round based rewards $\alpha(s, H)=\alpha(s, H|_r)$.  Similarly for $\alpha^{**}$, we have $\alpha^{**}(s, H)=\alpha^{**}(s, H|_r)$.  Hence, it suffices to prove $\alpha(s, H|_r)=\alpha^{**}(s, H|_r)$ at the restricted histories, i.e.,  histories with only a single round.

	Note that by Lemma~\ref{lemma2}, if the rounds in any two restricted histories are of the same size, then the shares are awarded based on their ranks. Hence, without loss of generality, we can denote these rounds $P_r=\{s_1,s_2,\ldots, s_n\}$ and the restricted histories as $H^n=H|_r$. Therefore, it suffices to prove $\alpha(s_i, H^n)=\alpha^{**}(s_i, H^n)$ for all $n$ and for all $i\leq n$. 
	
	\textbf{Single-share round:} Let $H^1$ be any history with a single round with a single share $s$. As $\alpha$ satisfies fixed total rewards and round based rewards then by Remark \ref{remark1}, $\alpha(s, H^1)=R$. Setting $\varepsilon(1)=1$ we have: 
	\begin{equation}\label{singe_share_round}
		\alpha(s,H^1)=\alpha^*(s,H^1)= R. 
	\end{equation}
	
	\textbf{Multi-share round:} Let $H^n$ be any history with a single round with multiple shares $s_1, \ldots, s_n$. In what follows, we show  that $\alpha(s,H^n)=\alpha^*(s,H^n)$ any $s\in P_r$, by induction on $n$, i.e., the size of the round $P_r$.

	\textbf{Induction Basis:} Let $n=2$. Let $H^2=(\{s_1,s_2\},\big\{\{s_1,s_2\}\big\} )$. We will show that  $\alpha(s_1, H^2)=\alpha^*(s_1, H^2)$ and $\alpha(s_2, H^2)=\alpha^*(s_2, H^2)$. By Lemma \ref{lemma1} for two histories $H^1$ and $H^2$, and by Equation~\ref{singe_share_round}, we have $\alpha(s_1, H^2)+\alpha(s_2, H^2)=R$. So, $\alpha(s_1, H^2)=R(1-\delta_2)$ and $\alpha(s_2, H^2)=R\delta_2$ for some $\delta_2 \in [0,1]$. Setting $\varepsilon(2)=\delta_2$, yields $\alpha(s_1, H^2)=R(1-\varepsilon(2))$ and $\alpha(s_2, H^2)=R\varepsilon(2)$.
	Therefore, $\alpha(s_1, H^2)=R\varepsilon(1)(1-\varepsilon(2))=\alpha^*(s_1, H^2)$ and $\alpha(s_2, H^2)=R\varepsilon(2)=\alpha^*(s_2, H^2)$. Note that, $\varepsilon(1)\in [0,1]$, and as $\delta_2\in [0,1]$ then $\varepsilon(2)\in [0,1]$.
	
	\textbf{Induction Hypothesis:} Let $n=k$ with $k>1$. Suppose we have $\alpha(s_i, H^k)=\alpha^*(s_i, H^k)$ for all $i\leq k$. 
	
	To prove for $n=k+1$, consider any $H^{k+1}=(\big\{s_1,\dots,s_k,s_{k+1}\big\}, \big\{\{s_1,\dots,s_k,s_{k+1}\}\big\})$. We will show that  $\alpha(s_i, H^{k+1})=\alpha^*(s_i, H^{k+1})$ for all  $i\leq k+1$. Let $H^k=(\big\{s_1,\dots,s_k\big\}, \big\{\{s_1,\dots,s_k\}\big\})$.

	By relative redistribution for all $i,j \in \{1,\dots, k\}$ we have
	\begin{equation}\label{relative_eq1}
		\frac{ \alpha(s_i , H^{k+1})}{\alpha(s_i , H^{k})}=\frac{ \alpha(s_j , H^{k+1})}{\alpha(s_j , H^{k})}
	\end{equation}
	Let us denote this ratio above by $1-\delta_{k+1}$. Therefore, $\alpha(s_i , H^{k+1})=(1-\delta_{k+1})\alpha(s_i , H^{k})$ for all $i\leq k$. By induction hypothesis 
	$\alpha(s_i, H^k)=\alpha^*(s_i, H^k)$ for all $i\leq k$ which implies $\alpha(s_i, H^{k+1})=(1-\delta_{k+1})\alpha^*(s_i, H^k)$. Therefore,
	
	\begin{equation}\label{equation_1_to_k}
		\alpha(s_i, H^{k+1})=(1-\delta_{k+1})R\varepsilon(i)\prod\limits_{j=i+1}^{k}(1-\varepsilon(j))\text{ for all } i\leq k
	\end{equation}

	\noindent By Equation~\ref{relative_eq1}, $\alpha(s_i , H^{k+1})=(1-\delta_{k+1})\alpha(s_i , H^{k})$, and by Lemma~\ref{lemma1} (on $H^k$ and on $H^1$), we have 
	\begin{equation}\label{equationsomething}
		\sum\limits_{i=1}^{k}\alpha(s_i , H^{k+1})=(1-\delta_{k+1})\sum\limits_{i=1}^{k}\alpha(s_i , H^{k})=(1-\delta_{k+1})R
	\end{equation}

	\noindent Note that 
	$ \sum\limits_{i=1}^{k+1}\alpha(s_i , H^{k+1})=\sum\limits_{i=1}^{k}\alpha(s_i , H^{k+1})+\alpha(s_{k+1},H^{k+1})$. Plugging Equation~\ref{equationsomething} into this, yields $\sum\limits_{i=1}^{k+1}\alpha(s_i , H^{k+1})=(1-\delta_{k+1})R+\alpha(s_{k+1} , H^{k+1})$. By Remark~\ref{remark1}, $\sum\limits_{i=1}^{k+1}\alpha(s_i , H^{k+1})=R$, therefore we have:
	\begin{equation}\label{equation_k+1}
		\alpha(s_{k+1} , H^{k+1})= \delta_{k+1}R.
	\end{equation}
	
	Setting $\varepsilon(k+1)=\delta_{k+1}$ in Equations~\ref{equation_1_to_k} and~\ref{equation_k+1}, proves $\alpha(s_i, H^{k+1})=\alpha^*(s_i, H^{k+1})$ for all $i\in \{1, \ldots, k\}$ and for $i=k+1$, respectively. 
	
	Note that,	 $\varepsilon(k+1)\in [0,1]$ for all $k>0$. Otherwise, if  $\varepsilon(k+1)<0$ then $\delta_{k+1}<0$. However, by Equation~\ref{equation_k+1}, this implies  $\alpha(s_{k+1} , H^{k+1})<0$ which contradicts Definition~\ref{def_reward_rule}. 
	If  $\varepsilon(k+1)>1$ then $\delta_{k+1}>1$ which implies $1-\delta_{k+1}<0$. However, by Equation~\ref{relative_eq1}, $\alpha(s_i , H^{k+1})=(1-\delta_{k+1})\alpha(s_i , H^{k})$ for all $i\leq k$, this implies  $\alpha(s_i , H^{k+1})<0$ for all $i\in \{1,\dots,k\}$ which contradicts Definition~\ref{def_reward_rule}. 
\end{proof}

\section{Proof of Theorem \ref{thm3}}\label{thm3_proof}
\textbf{Theorem 3.} 	A reward sharing scheme satisfies round based rewards, budget limit,  fixed total reward, ordinality, absolute redistribution, and relative redistribution if and only if it is $k$-pseudo proportional in the sense of ~\ref{pseudo_prop_def}. 
\begin{proof}
	Take any reward sharing scheme $\alpha$ that satisfies the axioms. Take any history $H=(S,\mathcal{P}(S))$ and any $s\in S$. Let $s\in P_r$ for some $1\leq r\leq l=|\mathcal{P}(S)|$. As $\alpha$ satisfies round based rewards $\alpha(s, H)=\alpha(s, H|_r)$.  Similarly for $\alpha^{k,\delta}$, we have $\alpha^{k,\delta}(s, H)=\alpha^{k,\delta}(s, H|_r)$.  Hence, it suffices to prove $\alpha(s, H|_r)=\alpha^{k,\delta}(s, H|_r)$ at the restricted histories, i.e.,  histories with only a single round.

	Note that by Lemma~\ref{lemma2}, if the rounds in any two restricted histories are of the same size, then the shares are awarded based on their ranks. Hence, without loss of generality, we can denote these rounds $P_r=\{s_1,s_2,\ldots, s_n\}$ and the restricted histories as $H^n=H|_r$. Therefore, it suffices to prove that there exist some $k$ and some $\delta$ such that $\alpha(s_i, H^n)=\alpha^{k,\delta}(s_i, H^n)$ for all $n$ and for all $i\leq n$. First of all, note that for all single share histories, $H^1$, Remark~\ref{remark1} implies that $\alpha$ coincides with $\alpha^{k,\delta}$ regardless of the choice of $k$ and $\delta$. So let $n\geq 2$.

	In what follows, we find -by iterating on $n$- that there exist $k$ and $\delta$ such that $\alpha(s_i, H^n)=\alpha^{k,\delta}(s_i,H^n)$ for all $n\geq 2$ and for all $i\leq n$. At each step, we ask if $\alpha$ distributes the awards proportionally, or not. The proof structure is as follows: 
	
	\begin{enumerate}
		\item At step $h$, if $\alpha$ distributes the awards proportionally, then we move to step $h+1$.
		\item At step $h$, if $\alpha$ distributes the awards disproportionately (while it was proportional at step $h-1$), then we set $k=h$, and $\delta=\alpha(s_h,H^h)$, i.e., the award of the last share in the round. Thereafter we show that the scheme $\alpha$ coincides with $\alpha^{k,\delta}$ for all possible rounds and all shares in these rounds. 
	\end{enumerate}
	Note a round of size $n=2$ is a special case, since a round of size $1$ cannot be decided to be proportional or disproportionate. Therefore we first treat such histories.	
	
	
	%
	\noindent \textbf{STEP 2}: Let $n=2$, and $H^2=(\{s_1,s_2\},\big\{\{s_1,s_2\}\big\} )$. By Remark \ref{remark1}, we have $\alpha(s_1, H^2)+\alpha(s_2, H^2)=R$. So $\Big(\alpha(s_1,H^2),\alpha(s_2,H^2)\Big)=(R-\gamma_2,\gamma_2)$ for some $\gamma_2\in [0,R]$. Note also that $\alpha^{k,\delta}(s_1, H^2)+\alpha^{k,\delta}(s_2, H^2)=R$ for any $k$ and for any $\delta$. Now, there are two cases, either the awards are proportional, or not.
	%
	%
	
	\textbf{Case 1}: If $R-\gamma_2=\gamma_2$, then $k>2$, therefore continue to next step, (Step 3, i.e., $n=3$).
	
	\textbf{Case 2}: If $R-\gamma_2\neq\gamma_2$, then we set $k =2$ and $\delta=\gamma_2$. Note that $\alpha^{2,\gamma_2}(s_1, H^2)= \frac{R-\gamma_2}{2-1}$ and $\alpha^{2,\gamma_2}(s_2, H^2)= \gamma_2$. 
	Hence, for $n= 2$, we have $\alpha(s_i, H^n)=\alpha^{2,\gamma_2}(s_i,H^n)$ for all $i\leq n$. 
	Next we also show, for any $n>2$, we have $\alpha(s_i, H^n)=\alpha^{2,\gamma_2}(s_i,H^n)$  for all $i\leq n$. 
	
	\textbf{Case 2a.} If $\gamma_2=R$, then $\Big(\alpha(s_1,H^2),\alpha(s_2,H^2)\Big)=(0,R)$. Consider any extension $H^3$ of $H^2$. As $\alpha$ does not assign negative awards, Remark \ref{remark1} and absolute redistribution implies $$\Big(\alpha(s_1,H^3),\alpha(s_2,H^3),\alpha(s_3,H^{3})\Big)=(0,R, 0).$$ Similar argument can be extended to $H^4$ and further, e.g., $(0,R,0,0,\ldots,0)$ which shows $\alpha(s_i, H^n)=\alpha^{2,\gamma_2}(s_i,H^n)$.
	
	\textbf{Case 2b.} If $\gamma_2=0$, then $\Big(\alpha(s_1,H^2),\alpha(s_2,H^2)\Big)=(R,0)$. Consider any extension $H^3$ of $H^2$. As $\alpha$ does not assign negative awards, Remark \ref{remark1} and absolute redistribution implies $$\Big(\alpha(s_1,H^3),\alpha(s_2,H^3),\alpha(s_3,H^{3})\Big)=(R,0, 0).$$ Similar argument can be extended to $H^4$ and further, e.g., $(R,0,0,0,\ldots,0)$ which shows $\alpha(s_i, H^n)=\alpha^{2,\gamma_2}(s_i,H^n)$.
	
	\textbf{Case 2c.} If $\gamma_2\in (0,R)$, then $\Big(\alpha(s_1,H^2),\alpha(s_2,H^2)\Big)=(R-\gamma_2,\gamma_2)$. Consider any extension $H^3$ of $H^2$. By absolute redistribution we have $\alpha(s_1,H^{2})-\alpha(s_1,H^{3})=\alpha(s_2,H^{2})-\alpha(s_2,H^{3})$. By relative redistribution we  have
	$$\frac{\alpha(s_1,H^{3})}{\alpha(s_1,H^{2})}=\frac{\alpha(s_2,H^{3})}{\alpha(s_2,H^{2})}=\theta$$
	Combining these equations, we have $\alpha(s_1,H^{2})-\alpha(s_1,H^2)\theta=\alpha(s_2,H^2)-\alpha(s_2,H^2)\theta$ which implies $\alpha(s_1,H^2)(1-\theta)=\alpha(s_2,H^2)(1-\theta)$. As $R-\gamma_2\neq\gamma_2$, the previous equation only holds if $\theta=1$. This results in $\alpha(s_1,H^{3})=\alpha(s_1,H^{2})$ and $\alpha(s_2,H^{3})=\alpha(s_2,H^{2})$. Finally,  as $\alpha$ does not assign negative awards,  Remark~\ref{remark1} implies
	$$\Big(\alpha(s_1,H^3),\alpha(s_2,H^3),\alpha(s_3,H^{3})\Big)=(R-\gamma_2,\gamma_2, 0).$$
	Similar argument can be extended to $H^4$ and further, e.g., $(R-\gamma_2,\gamma_2,0,0,\ldots,0)$ which shows $\alpha(s_i, H^n)=\alpha^{2,\gamma_2}(s_i,H^n)$.

	\noindent	\textbf{STEP h:} Let $n=h$, and $H^h=(\{s_1,\dots,s_h\},\big\{\{s_1,\dots,s_h\}\big\} )$.  Reaching to step $h$ implies the awards to shares at $H^{h-1}=(\{s_1,\dots,s_{h-1}\},\big\{\{s_1,\dots,s_{h-1}\}\big\} )$ were distributed proportionally and hence all are equal. 	By relative redistribution for all $i,j\leq h-1$ we have, 
	$\frac{\alpha(s_i,H^{h})}{\alpha(s_i,H^{h-1})}=\frac{\alpha(s_j,H^{h})}{\alpha(s_j,H^{h-1})}$. 
	This implies  $\alpha(s_i,H^{h})=\alpha(s_j,H^{h})$ for all $i,j\leq h-1$. By Remark \ref{remark1}, $\sum_{i=1}^{h}\alpha(s_i,H^h)=R$, which implies $(h-1)\alpha(s_i,H^h)+\alpha(s_h,H^h)=R$ for any $i\leq h-1$. Therefore, $\alpha(s_i,H^h)=\frac{R-\alpha(s_h,H^h)}{h-1}$. Let us denote $\alpha(s_h,H^h)=\gamma_h$ for some $\gamma_h\in [0,R]$, so  $\Big(\alpha(s_1,H^h),\alpha(s_2,H^h),\dots, \alpha(s_{h-1},H^{h}), \alpha(s_{h},H^{h})\Big)=(\frac{R-\gamma_h}{h-1}, \frac{R-\gamma_h}{h-1},\dots, \frac{R-\gamma_h}{h-1}, \gamma_h)$. Now, there are two cases, either the awards are proportional, or not.
	
	\textbf{Case 1}: If $\frac{R-\gamma_h}{h-1}=\gamma_h$, then $k>h$, therefore continue to next step, (Step $h+1$, i.e., $n=h+1$).
	
	\textbf{Case 2}: If $\frac{R-\gamma_h}{h-1}\neq \gamma_h$, then we set $k=h$ and $\delta=\gamma_h$. Note that $\alpha^{h,\gamma_h}(s_i, H^h)= \frac{R-\gamma_h}{h-1}$ for all $i< h$ and $\alpha^{h,\gamma_h}(s_h, H^h)= \gamma_h$. 
	Hence, for $n= h$, we have $\alpha(s_i, H^n)=\alpha^{h,\gamma_h}(s_i,H^n)$ for all $i\leq n$. 
	Next we also show, for any $n>h$, we have $\alpha(s_i, H^n)=\alpha^{h,\gamma_h}(s_i,H^n)$  for all $i\leq n$. 
	
	\textbf{Case 2a.} If $\gamma_h=R$, then $\Big(\alpha(s_1,H^h),\alpha(s_2,H^h), \dots, \alpha(s_{h-1},H^h), \alpha(s_h,H^h)\Big)=(0,\dots,0,R)$. Consider any extension $H^{h+1}$  of $H^h$.  As $\alpha$ does not assign negative awards, Remark \ref{remark1} and absolute redistribution implies 
	$$\Big(\alpha(s_1,H^{h+1}), \dots, \alpha(s_{h},H^{h+1}), \alpha(s_{h+1},H^{h+1})\Big)=(0,\dots,0,R,0)$$
	Similar argument can be extended to $H^{h+2}$ and further, e.g., $(0,\ldots,0,R,0,\dots,0)$ which shows $\alpha(s_i, H^n)=\alpha^{h,\gamma_h}(s_i,H^n)$.
	
	\textbf{Case 2b.} If $\gamma_h=0$, then $\Big(\alpha(s_1,H^h),\alpha(s_2,H^h), \dots, \alpha(s_{h-1},H^h), \alpha(s_h,H^h)\Big)=(\frac{R}{h-1},\dots,\frac{R}{h-1},0)$. Consider any extension $H^{h+1}$  of $H^h$.  As $\alpha$ does not assign negative awards, Remark \ref{remark1} and absolute redistribution implies 
	$$\Big(\alpha(s_1,H^{h+1}), \dots, \alpha(s_{h},H^{h+1}), \alpha(s_{h+1},H^{h+1})\Big)=(\frac{R}{h-1},\dots,\frac{R}{h-1},0,0)$$
	Similar argument can be extended to $H^{h+2}$ and further, e.g., $(\frac{R}{h-1},,\ldots,\frac{R}{h-1},0,\dots,0)$ which shows $\alpha(s_i, H^n)=\alpha^{h,\gamma_h}(s_i,H^n)$.
	
	\textbf{Case 2c.} If $\gamma_h\in (0,R)$, then $\Big(\alpha(s_1,H^h),\alpha(s_2,H^h), \dots, \alpha(s_{h-1},H^h), \alpha(s_h,H^h)\Big)=(\frac{R-\gamma_h}{h-1}, \allowbreak \frac{R-\gamma_h}{h-1}, \dots, \frac{R-\gamma_h}{h-1},\gamma_h)$. Consider any extension $H^{h+1}$  of $H^h$.
	By absolute redistribution  for all $i,j\leq h$ we have $\alpha(s_i,H^{h})-\alpha(s_i,H^{h+1})=\alpha(s_j,H^{h})-\alpha(s_j,H^{h+1})$.
	By relative redistribution  for all $i,j\leq h$ we  have
	$$\frac{\alpha(s_i,H^{h+1})}{\alpha(s_i,H^{h})}=\frac{\alpha(s_j,H^{h+1})}{\alpha(s_j,H^{h})}=\theta$$
	Combining these equations, we have $\alpha(s_i,H^{h})-\alpha(s_i,H^h)\theta=\alpha(s_j,H^h)-\alpha(s_j,H^h)\theta$ which implies $\alpha(s_i,H^h)(1-\theta)=\alpha(s_j,H^h)(1-\theta)$.
	As   $\alpha(s_i,H^h)\neq\alpha(s_j,H^h)$ for all $i,j\leq h$, the previous equation only holds if $\theta=1$. This results in $\alpha(s_i,H^{h+1})=\alpha(s_i,H^{h})$  for all $i\leq h$. Finally,  as $\alpha$ does not assign negative awards,  Remark~\ref{remark1} implies
	$$\Big(\alpha(s_1,H^{h+1}),\dots,\alpha(s_{h-1},H^{h+1}), \alpha(s_{h},H^{h+1}), \alpha(s_{h+1},H^{h+1})\Big)=(\frac{R-\gamma_h}{h-1},\dots,\frac{R-\gamma_h}{h-1}, \gamma_h,0).$$
	Similar argument can be extended to $H^4$ and further, e.g., $(\frac{R-\gamma_h}{h-1},\dots,\frac{R-\gamma_h}{h-1}, \gamma_h,0,\dots,0)$ which shows $\alpha(s_i, H^n)=\alpha^{2,\gamma_2}(s_i,H^n)$.
	
	Note that in case $\alpha$ never distributes the awards ``disproportionately'', then this implies that $h$ goes to infinity and therefore, we set $k=\infty$ and it is clear to see that $\alpha=\alpha^{\infty,\delta}$ for any $\delta$, i.e., $\alpha$ is the proportional scheme, which is an element of $k$-pseudo proportional class. All in all, this completes the proof.
\end{proof}

\section{Known reward sharing schemes}\label{formal_rules}
In this section we focus on two of the most popular, and widely applied, reward sharing schemes and examine whether they satisfy the axioms proposed in Section \ref{axioms_section}.  We also comment on several potentially interesting schemes suggested in \cite{rosenfeld2011analysis} and \cite{schrijvers2016incentive} respectively.

$\bullet$  \textbf{Pay Per Share (PPS):}
The strict egalitarian Pay-Per-Share scheme fails to ensure an economically viable mining pool (as noted in \cite{rosenfeld2011analysis}) but has been applied, for instance by F2Pool\footnote{\url{https://www.f2pool.com/}} and Poolin\footnote{\url{https://www.poolin.com/}}, probably due to its immediate simplicity and transparency.  
In pay per share, every submitted share receives a fixed  reward regardless of when it is submitted, and the round it is submitted in. Formally, $\alpha(s, H)= c$ for some constant $c$. The payments to the shares are usually adjusted by the network difficulty and the length of a round.
Trivially, PPS fails fixed total reward and budget limit, but satisfies relative redistribution, absolute redistribution, ordinality and round based rewards.  

$\bullet$ \textbf{Pay-Per-Last-N-Shares:} The pay-per-last-N-shares (PPLNS for short) has also been popular and is applied, for instance in GHash.IO\footnote{\url{https://ghash.io/}} and P2Pool\footnote{\url{http://p2pool.in/}}. It has been claimed that this scheme (under certain conditions) prevents miners from delay reporting their shares (see \cite{schrijvers2016incentive} and \cite{lazos2021rpplns} for a modified version).
In PPLNS, the pool manager gets a fixed fee, say $f$, and the net-reward, $R=(1-f)B$, is distributed (proportionally) among the $N$ last shares (including the full the solution), regardless of the round boundaries. Therefore, the reward of a share at the time it is submitted depends on the number of full solutions among the next $N-1$ shares. That is, if no full solution is found among the next $N-1$ shares, the share receives no reward, if one full solution is found it is rewarded once with $\frac{1}{N}R$, if two  full solutions are found it is rewarded twice with $\frac{1}{N}R$, and so forth.

To formally define the PPLNS, consider a history $H=(S, \mathcal{P}(S))$ and let $\Omega(s)$  denote the index of the pool round that share $s$ belongs to i.e., $\Omega(s)=\{x\leq |\mathcal{P}(S)| \mid s \in P_x\}$. Then the PPLNS scheme is defined as
$$\alpha(s_i, H)= \frac{\Omega(s_{N+i})-\Omega(s_i)}{N}R.$$

It is trivial that PPLNS satisfies ordinality. However, Example \ref{exmp4} shows that PPLNS fails to satisfy budget limit, fixed total rewards, round based rewards, absolute redistribution and relative redistribution. 

\begin{example}\label{exmp4}
	Consider the PPLNS reward sharing scheme with $N=3$. Let $H=(S, \mathcal{P}(S))$ be a history as follows,
	\begin{align*}
		S&=\{s_1,s_2,\dots, s_{1000}\}\\
		\mathcal{P}(S)&=\{ \{\underbrace{s_1,s_2,s_3,s_4,\mathbf{s_5}}_{P_1}\}, \{\underbrace{\mathbf{s_6}}_{P_2}\},
		\{\underbrace{\mathbf{s_7}}_{P_3}\},
		\{ \underbrace{\mathbf{s_8}}_{P_4} \},
		\{ \underbrace{s_9,\dots,\mathbf{s_{20}}}_{P_5}\},\dots, \{\underbrace{\dots, \mathbf{s_{1000}}}_{P_{132}}\} 
		\}
	\end{align*}
	
	To show that PPLNS fails to satisfy the budget limit and  fixed total rewards, note that the rewards of the shares in the first and second rounds of the aforementioned history are 
	\begin{align*}
		\alpha(s_1, H)&=0 &\alpha(s_2, H)&=0  & \alpha(s_3, H)&=\frac{1}{3}R   \\
		\alpha(s_4, H)&=\frac{2}{3}R  &  \alpha(s_5, H)&=R    & \alpha(s_6, H)&=\frac{3}{3}R
	\end{align*}
	Therefore, as   $\sum\limits_{s\in P_1}\alpha(s, H)=\frac{6}{3}R$, the PPLNS  fails to satisfy the budget limit. Also, as $\sum\limits_{s\in P_1}\alpha(s, H) \neq \sum\limits_{s\in P_2}\alpha(s, H)$  the PPLNS  fails to satisfy the fixed total rewards.
	
	To show that PPLNS fails to satisfy the round based rewards, consider the restriction of $H$ to the first round, i.e., $H|_1 = \Big( \{s_{1},\dots,\mathbf{s_{5}}\}, \, \Big\{\{\underbrace{s_{1},\dots,\mathbf{s_{5}}}_{P_1}\}\Big\} \Big)$. The rewards of each share would be $ \alpha(s_1, H|_1)= \alpha(s_2, H|_1)=0$, and $\alpha(s_3, H|_1)= \alpha(s_4, H|_1)=\alpha(s_5, H|_1)=\frac{R}{3}$. Comparing these with the reward of each share at the history $H$ shows that the PPLNS  fails to satisfy the round based rewards.
	
	To show that PPLNS fails to satisfy absolute redistribution and relative redistribution, consider $H'= (S',\mathcal{P}'(S'))$ as an extension of $H$ at the first round as follows:
	\begin{align*}
		S'&=\{s_1,s_2,\dots, s^*, \dots, s_{1000}\}\\
		\mathcal{P}'(S')&=\{ \{\underbrace{s_1,s_2,s_3,s_4,s_5,\mathbf{s^*}}_{P_1}\}, \{\underbrace{\mathbf{s_6}}_{P_2}\},
		\{\underbrace{\mathbf{s_7}}_{P_3}\},
		\{ \underbrace{\mathbf{s_8}}_{P_4} \},
		\{ \underbrace{s_9,\dots,\mathbf{s_{20}}}_{P_5}\},\dots, 	 \{\underbrace{\dots, \mathbf{s_{1000}}}_{P_{132}}\} 
		\}
	\end{align*}
	It is easy to verify that $\alpha(s_1, H')=\alpha(s_2, H')=\alpha(s_3, H')=0$, $\alpha(s_4, H')=\frac{1}{3}R$, $\alpha(s_5, H')=\frac{2}{3}R$, and $\alpha(s^*, H')=R$.
	
	Since, $\alpha(s_1, H)-\alpha(s_1, H')\neq\alpha(s_3, H)-\alpha(s_3, H')$ then the PPLNS fails to satisfy the absolute redistribution. Also as $\frac{\alpha(s_4, H')}{\alpha(s_4, H)}\neq\frac{\alpha(s_5, H')}{\alpha(s_5, H)}$ the PPLNS fails to satisfy the relative redistribution.
\end{example}
Besides the two popular schemes analyzed above, the academic literature has suggested additional and more sophisticated schemes designed to provide miners with improved incentives.  For instance, \cite{rosenfeld2011analysis} suggests the to use a so-called geometric scheme.

$\bullet$ \textbf{Geometric}: In this scheme, the rewards are distributed among the shares in a round using a geometric series based on the order of their submission. Unlike other schemes the fees are variable in this model and they depend on the size of the round. Formally, let $r>1$. The fee for a round  is defined as $f(P(s))=\frac{1}{r^{|P(s)|}}$ and the reward to each share is 	
$$\alpha(s, H)=\frac{(r-1)}{r^{|P(s)|-\rho(s)+1}}B.$$

It is straight forward to see that the geometric scheme satisfies relative redistribution, round based rewards and ordinality. The following proposition shows it also satisfies  budget limit.
\begin{prop}
	The geometric scheme satisfies budget limit.
\end{prop}
\begin{proof} Let $H=(S,\mathcal{P}(S))$ be any history. Take any round $r$. Let  $|P_r|=k$. Without loss of generality, assume $P_r=\{1,2,\dots,k\}$ so that $\rho(s)=s$. Then
	\begin{align}\label{eq1}
		\sum_{s=1}^{k}\alpha(s, H)
		&=\sum_{s=1}^{k}\frac{(r-1)}{r^{k-s+1}}B =B(r-1)\sum\limits_{s=1}^{k}\frac{1}{r^s}\nonumber\\
		& =B(r-1)\frac{1}{r}\left(\frac{1-\frac{1}{r^n}}{1-\frac{1}{r}}\right)  =B(r-1)\left(\frac{1-\frac{1}{r^n}}{r-1}\right) =B(1- \frac{1}{r^n}) 
	\end{align}
	As $r>1$ then Equation  \ref{eq1}  is always less than $B$.
\end{proof}

The following example shows the geometric scheme fails to satisfy the  fixed total reward and absolute redistribution.

\begin{example}\label{exmp5}
	Consider a  history $H=(S, \mathcal{P}(S))$ as follows,
	\begin{align*}
		S&=\{s_1,s_2,\dots, s_{1000}\}\\
		\mathcal{P}(S)&=\{ \{\underbrace{s_1,\mathbf{s_2}}_{P_1}\}, \{\underbrace{\mathbf{s_3}}_{P_2}\},
		\{ \underbrace{s_4,\dots,\mathbf{s_{20}}}_{P_3}\},\dots, \{\underbrace{\dots, \mathbf{s_{1000}}}_{P_{132}}\} 
		\}
	\end{align*}
	It is easy to verify that $
	\alpha(s_1, H)=\frac{r-1}{r^2}B$, $\alpha(s_2, H)=\frac{r-1}{r}B$, and $\alpha(s_3, H)=\frac{r-1}{r}B$.
	
	As   $\sum\limits_{s\in P_1}\alpha(s, H) \neq \sum\limits_{s\in P_2}\alpha(s, H)$ then geometric scheme  fails to satisfy the fixed total rewards.
	
	To show that geometric fails to satisfy absolute redistribution consider $H'= (S',\mathcal{P}'(S'))$ as an extension of $H$ at the first round as follows:
	\begin{align*}
		S'&=\{s_1,s_2,\dots, s^*, \dots, s_{1000}\}\\
		\mathcal{P}'(S')&=\{ \{\underbrace{s_1,s_2,\mathbf{s^*}}_{P_1}\}, \{\underbrace{\mathbf{s_3}}_{P_2}\},
		\{ \underbrace{s_4,\dots,\mathbf{s_{20}}}_{P_3}\},\dots, 	 \{\underbrace{\dots, \mathbf{s_{1000}}}_{P_{132}}\} 
		\}
	\end{align*}
	It is easy to verify that $\alpha(s_1, H')=\frac{r-1}{r^3}B$, $\alpha(s_2, H')=\frac{r-1}{r^2}B$, $\alpha(s^*, H')=\frac{r-1}{r}B$.	Since $r>1$ and $\alpha(s_1, H)-\alpha(s_1, H')\neq\alpha(s_2, H)-\alpha(s_2, H')$ then the geometric fails to satisfy the absolute redistribution.
\end{example}

Note that the Geometric scheme fails to satisfy the fixed total reward axiom.  We propose the  following modification of this scheme to fix this failure.  

$\bullet$ \textbf{Constrained Geometric:} This scheme is defined as $\alpha(s, H)=\dfrac{(r-1)}{r^{|P(s)|-\rho(s)+1}}\times \dfrac{r^{|P(s)|}}{r^{|P(s)|}-1}$. The constrained geometric scheme is included in the class of relative fairness reward sharing schemes (with $\epsilon_1=1$ and $\epsilon_j=\frac{r^{j-1}-1}{r^{j}-1}$). Therefore, it satisfies budget limit, total fix rewards, ordinality and relative redistribution. To show that it fails absolute re-distribution consider the history $H$ and its extension $H'$ as presented in Example \ref{exmp5}. It can be verified that $
\alpha(s_1, H)=\frac{r-1}{r^2}\times \frac{r^2}{r^2-1}B$, $\alpha(s_2, H)=\frac{r-1}{r}\times\frac{r^2}{r^2-1}B$, and $\alpha(s_1, H')=\frac{r-1}{r^3}\times\frac{r^3}{r^3-1}B$, $\alpha(s_2, H')=\frac{r-1}{r^2}\times\frac{r^3}{r^3-1}B$, $\alpha(s^*, H')=\frac{r-1}{r}\times\frac{r^3}{r^3-1}B$. Since $r>1$ and $\alpha(s_1, H)-\alpha(s_1, H')\neq\alpha(s_2, H)-\alpha(s_2, H')$ then the  constrained geometric fails to satisfy the absolute redistribution.

Next, we formulate another scheme which approaches the reward sharing problem from  the  ``incentive compatibility'' perspective, i.e., it provides miners with the incentive to report shares immediately.

$\bullet$ {\bf IC scheme:} The IC scheme is proposed by \cite{schrijvers2016incentive} as an \textit{incentive compatible} reward sharing scheme. In words, let $1/D$ denote the probability of a share to be a full solution. Then, if the round length is at least of the same size as $D$, the scheme distributes the reward proportionally according to the length of the round; if the round length is shorter, every share receives $1/D$ and the residual budget is given to the last (full) share of the round. Formally,
\begin{equation*}
	\alpha(s,H) = \left\{\begin{array}{ll}
		\frac{R}{D} & \textit{~~ if ~~}  |P(s)|\leq D \text{ and } \rho(s)<|P(s)|\\
		\frac{R}{D}+ (1-\frac{|P(s)|}{D})R, & \textit{~~ if ~~}  |P(s)|\leq D \text{ and } \rho(s)=|P(s)| \\
		\frac{R}{|P(s)|}, & \textit{~~ if ~~}  |P(s)|\geq D 
	\end{array}\right.
\end{equation*}
This rule obviously satisfies fixed total rewards, ordinality, and round based rewards, while it fails both absolute and relative redistribution.

We conclude this section by formulating with one more reward sharing scheme, which is interesting. The Slush scheme, named after the Slush mining pool\footnote{\url{https://slushpool.com/}}, is the only reward sharing scheme that uses time signatures as a parameter for distributing the rewards. Our framework is rich enough to capture this feature. Formally: 

$\bullet$  \textbf{Slush:}  Consider a history $H=(S, \mathcal{P}(S))$ and let $\Omega(s)$  denote the index of the pool round that share $s$ belongs to i.e., $\Omega(s)=\{x\leq |\mathcal{P}(S)| \mid s \in P_x\}$. Let $\bar{s}_j$ denote the last share in the $j^{th}$ round, i.e., $\bar{s}_j=\{s\in P_j\mid \tau(s) \geq \tau(s'),~\forall~s'\in P_j \}$. Then

Let $score(s, j)= \dfrac{e^\frac{{\tau(s)-\tau(\bar{s}_{j} )}}{\lambda}
}{\sum\limits_{\tau(s')\leq \tau(\bar{s}_{j})} e^\frac{{\tau(s')-\tau(\bar{s}_{j})}}{\lambda}}$ for any $\Omega(s) \leq j\leq l$. Then
$$\alpha(s, H)=R\sum\limits_{i=\Omega(s)}^{l}score(s,i).$$ 

The parameter $\lambda$  is set to $1200$ in the Slush pool. In what follows, we therefore assume $\lambda=1200$.

In the following example we show that the slush scheme does not satisfy fixed total reward, ordinality, budget limit, round based reward, absolute redistribution and relative redistribution.

\begin{example}
	Consider a  history $H=(S, \mathcal{P}(S))$ as follows,
	\begin{align*}
		S&=\{s_1,s_2, s_3\}\\
		\mathcal{P}(S)&=\{ \{\underbrace{s_1,\mathbf{s_2}}_{P_1}\},
		\{\underbrace{\mathbf{s_3}}_{P_2}\}
		\}
	\end{align*}
	Let $\tau(s_1)=1, \tau(s_2)=2$ and $\tau(s_1)=3$. Therefore, we have:
	\begin{align*}
		& score(s_1,1)=\frac{e^{\frac{1-2}{1200}}}{e^{\frac{1-2}{1200}}+e^{\frac{2-2}{1200}}}\approx 0.49\\
		& score(s_1,2)=\frac{e^{\frac{1-3}{1200}}}{e^{\frac{1-3}{1200}}+e^{\frac{2-3}{1200}}+e^{\frac{3-3}{1200}}}\approx 0.33\\
		&	  score(s_2,1)=\frac{e^{\frac{2-2}{1200}}}{e^{\frac{1-2}{1200}}+e^{\frac{2-2}{1200}}}\approx 0.5\\
			&   score(s_2,2)= \frac{e^{\frac{2-3}{1200}}}{e^{\frac{1-3}{1200}}+e^{\frac{2-3}{1200}}+e^{\frac{3-3}{1200}}}\approx 0.33\\
	\end{align*}
		\begin{align*}			
		&   score(s_3,2)=\frac{e^{\frac{3-3}{1200}}}{e^{\frac{1-3}{1200}}+e^{\frac{2-3}{1200}}+e^{\frac{3-3}{1200}}}\approx 0.33\\
	\end{align*}
	Therefore,
	$\alpha(s_1, H)=0.82R$, $\alpha(s_2, H)=.83R$ and $\alpha(s_3, H)=.33R$. 
	
	As   $\sum\limits_{s\in P_1}\alpha(s, H)=1.65R$, the Slush scheme violates the budget limit. As $\sum\limits_{s\in P_1}\alpha(s, H) \neq \sum\limits_{s\in P_2}\alpha(s, H)$, it fails to satisfy the fixed total rewards. It is easy to see that the Slush scheme also fails the round based reward axiom, e.g., for the second round in this example. In addition, one can verify that changing the time signature of any of the shares will have an effect on the award of shares, therefore the Slush scheme also violates ordinality.
	
	Now consider $H'= (S', \mathcal{P}'(S'))$ as an extension of $H$ at the first round as follows:
	\begin{align*}
		S'&=\{s_1,s_2, s^*, s_{3}\}\\
		\mathcal{P}'(S')&=\{ \{\underbrace{s_1,s_2,\mathbf{s^*}}_{P_1}\}, \{\underbrace{\mathbf{s_3}}_{P_2}\}
		\}
	\end{align*}
	such that $\tau(s^*)=2.5$	Then the award of each share would be:
	\begin{align*}
		& score(s_1,1)=\frac{e^{\frac{1-2.5}{1200}}}{e^{\frac{1-2.5}{1200}}+e^{\frac{2-2.5}{1200}}+e^{\frac{2.5-2.5}{1200}}}\approx 0.33\\
		& score(s_1,2)=\frac{e^{\frac{1-3}{1200}}}{e^{\frac{1-3}{1200}}+e^{\frac{2-3}{1200}}+e^{\frac{2.5-3}{1200}}+e^{\frac{3-3}{1200}}}\approx 0.25\\
		&	  score(s_2,1)=\frac{e^{\frac{2-2.5}{1200}}}{e^{\frac{1-2.5}{1200}}+e^{\frac{2-2.5}{1200}}+e^{\frac{2.5-2.5}{1200}}}\approx 0.33 \\
		& 
		score(s_2,2)= \frac{e^{\frac{2-3}{1200}}}{e^{\frac{1-3}{1200}}+e^{\frac{2-3}{1200}}+e^{\frac{2.5-3}{1200}}+e^{\frac{3-3}{1200}}}\approx 0.25\\
			&   score(s^*,1)= \frac{e^{\frac{2.5-2.5}{1200}}}{e^{\frac{1-2.5}{1200}}+e^{\frac{2-2.5}{1200}}+e^{\frac{2.5-2.5}{1200}}}\approx 0.33\\
		&   score(s^*,2)= \frac{e^{\frac{2.5-3}{1200}}}{e^{\frac{1-3}{1200}}+e^{\frac{2-3}{1200}}+e^{\frac{2.5-3}{1200}}+e^{\frac{3-3}{1200}}}\approx 0.25 \\
		&   score(s_3,2)=\frac{e^{\frac{3-3}{1200}}}{e^{\frac{1-3}{1200}}+e^{\frac{2-3}{1200}}+e^{\frac{2.5-3}{1200}}+e^{\frac{3-3}{1200}}}\approx 0.25 \\
	\end{align*}
	Therefore,
	$\alpha(s_1, H')=0.58R$ and $\alpha(s_2, H')=.58R$. Comparing these to $\alpha(s_1, H)=0.82R$ and $\alpha(s_2, H)=.83R$, shows that the Slush scheme fails both absolute and fair redistribution axioms.
\end{example}

The results for the well-known schemes are summarized in Table~\ref{tab2} below.

\begin{table}[h]
	\centering
	\caption{Summary of the well-known schemes.}\label{tab2}
	\small
	\begin{tabular}{ >{\raggedright}p{3.1cm} 
			>{\centering}p{1.7cm}
			>{\centering}p{1.8cm} >{\centering}p{1.9cm} >{\centering}p{2cm} >{\centering}p{1.1cm}c}
		\hline
		Scheme & Fixed total reward & Relative redistribution & Absolute redistribution& Round based rewards& Budget limit & Ordinality \\ \hline 
		PPS & - & +& +  & + & - & +\\
		PPLNS & - & - & -  & - & - & + \\
		Geometric & - & +& -  & +& + & + \\
		Modified Geometric & + & +& -  & + & +& + \\
		IC & + & -& -  & + & + & +\\
		Slush & - & -& -  & - & - & -\\
		\hline
	\end{tabular}
\end{table}

\section{Logical Independence}\label{independence_section}
We define six reward sharing schemes in order to demonstrate logical independence of the axioms in Section \ref{axioms_section}. The results are summarized in Table~\ref{tab1} below. Defining schemes 1-6, let $R=B-f$ for some fixed $f\in[0,B]$.
\begin{itemize}
	\item \textbf{Scheme 1:} \[
	\alpha(s,H) = \left\{\begin{array}{ll}
		\frac{R}{|P(s)|}, & \text{ if } |P(s)| \text{ is odd} \\
		\frac{R}{2 |P(s)|}, & \text{ if } |P(s)| \text{ is even}\\
	\end{array}\right.
	\]
	This scheme fails {\it fixed total rewards}, but meets all the other axioms. 
	\item \textbf{Scheme 2:}
	\[
	\alpha(s, H) = \left\{\begin{array}{ll}
		R, & \text{ if } |P(s)|=1\\
		(R-\lambda)+ \frac{\lambda}{|P(s)|}, & \text{ if } |P(s)|>1 \text{ and } \rho(s)= 1\\
		\frac{\lambda}{|P(s)|}, & \text{ if } |P(s)|>1 \text{ and } \rho(s)\neq 1\\
	\end{array}\right.
	\]
	where $0<\lambda< R$ is a constant number. 
	
	This scheme fails {\it relative redistribution}, but meets all the other axioms.
	\item \textbf{Scheme 3:}
	\[ \alpha(s,H) = \left\{\begin{array}{ll}
		\frac{R}{2^{|P(s)|-1}}, & \text{ if } \rho(s)= 1 \\
		\frac{2^{\rho(s)-2}R}{2^{|P(s)|-1}}, & \text{ if }\rho(s)\neq 1 
	\end{array}\right.
	\]
	This scheme fails {\it absolute redistribution}, but meets all the other axioms. 
	\item \textbf{Scheme 4:} \[
	\alpha(s,H) = \left\{\begin{array}{ll}
		\frac{R}{|P(s)|}, & \text{ if the number of shares in the first round of the history is odd} \\
		\frac{R}{2 |P(s)|}, & \text{ if the number of shares in the first round of the history is even} 
	\end{array}\right.
	\]
	This scheme fails {\it round based rewards}, but meets all the other axioms.
	\item	\textbf{Scheme 5:} \[
	\alpha(s,H) = \frac{2R}{|P(s)|}
	\] 
	
	This scheme fails {\it budget limit}, but meets all the other axioms. 
	\item \textbf{Scheme 6:}
	\[
	\alpha(s,H) = \left\{\begin{array}{ll}
		R, & \textit{~~ if ~~}  |P(s)|=1 \\
		\alpha^{2,\frac{R}{2}}(s,H), & \textit{~~ if ~~}  \tau(\rho(s)=1) - \tau(\rho(s)=2) < T \\
		\alpha^{2,\frac{R}{3}}(s,H), & \textit{~~ if ~~}  \tau(\rho(s)=1) - \tau(\rho(s)=2) \geq T 
	\end{array}\right.
	\]
	where $T$ is a time threshold.
	
	This scheme fails {\it ordinality}, but meets all the other axioms. 
	
\end{itemize}

\begin{table}[ht]
	\centering
	\caption{Logical independence of the axioms.}\label{tab1}
	\small
	\begin{tabular}{ >{\centering}p{1.5cm} >{\centering}p{2cm}
			>{\centering}p{2cm} >{\centering}p{2cm} >{\centering}p{2cm} >{\centering}p{2cm}ccccc}
		\hline
		Scheme & Fixed total reward & Relative redistribution & Absolute redistribution& Round based rewards& Budget limit & Ordinality \\ \hline 
		Scheme 1 & - & +& +  & + & + & +\\
		Scheme 2 & + & - & +  & + & + & + \\
		Scheme 3& + & +& -  & +& + & + \\
		Scheme 4 & + & +& +  & - & +& + \\
		Scheme 5 & + & +& +  & + & - & +\\
		Scheme 6 & + & +& +  & + & + & -\\
		\hline
	\end{tabular}
\end{table}

\hspace{15cm}~

\end{document}